\documentclass{article}

\RequirePackage{fontenc}
\RequirePackage{amsthm,amsmath}
\RequirePackage[numbers]{natbib}
\RequirePackage[colorlinks,citecolor=blue,urlcolor=blue]{hyperref}

\usepackage[utf8]{inputenc}
\usepackage{amssymb, mathrsfs}
\usepackage{graphicx,psfrag,epsf}
\usepackage{xcolor}
\usepackage{natbib}
\usepackage{geometry}
\usepackage{calc}
\usepackage{subfig}
\usepackage{enumerate}
\usepackage{url} 
\usepackage{soul} 
\usepackage{array}
\usepackage{cleveref} 

\numberwithin{equation}{section}
\theoremstyle{plain}
\newtheorem{thm}{Theorem}[section]
\newtheorem{korr}[thm]{Corollary}
\newtheorem{prop}[thm]{Proposition}
\newtheorem{lemma}[thm]{Lemma}

\newtheorem{defi}[thm]{Definition}
\newtheorem*{defi*}{Definition}
\newtheorem{remark}[thm]{Remark}

\newtheorem*{algorithm}{Algorithm}

\newcolumntype{L}{>{\centering\arraybackslash}m{1.5cm}}

\newcommand{\eb}{\Leftrightarrow}
\newcommand{\pil}{\rightarrow}

\newcommand{\R}{\mathbb{R}}

\newcommand{\fFDR}{\mathrm{fFDR}}

\newcommand{\for}{\text{ for }}

\newcommand{\tss}{\alpha} 
\newcommand{\dss}{\alpha} 
\newcommand{\bt}{\mathrm{\mathbf{T}}}

\newcommand{\E}{\mathrm{E}}

\newcommand{\de}{\: \mathrm{d}}

\title{False Discovery Rate for Functional Data}
\author{Niels Lundtorp Olsen\textsuperscript{1}, Alessia Pini\textsuperscript{2} and Simone Vantini\textsuperscript{3}}

\begin{document}

\vfill

\maketitle

\begin{abstract}
Since Benjamini and Hochberg introduced false discovery rate (FDR) in their seminal paper, this has become a very popular approach to the multiple comparisons problem. An increasingly popular topic within functional data analysis is local inference, i.e., the continuous statistical testing of a null hypothesis along the domain. The principal issue in this topic is the infinite amount of tested hypotheses, which can be seen as an extreme case of the multiple comparisons problem.

In this paper we define and discuss the notion of false discovery rate in a very general functional data setting.
Moreover, a continuous version of the Benjamini-Hochberg procedure is introduced along with a definition of adjusted p-value function. Some general conditions are stated, under which the functional Benjamini-Hochberg  procedure provides control of the functional FDR.
Two different simulation studies are presented; the first study has a one-dimensional domain and a comparison with another state of the art method, and the second study has a planar two-dimensional domain.

Finally, the proposed method is applied to satellite measurements of Earth temperature. In detail, we aim at identifying the regions of the planet where temperature has significantly increased in the last decades. After adjustment, large areas are still significant.
\end{abstract}

\thispagestyle{empty}

\vfill
\noindent
\textbf{1} Department of Mathematical Sciences, \textit{University of Copenhagen}
\\
\textbf{2} Department of Statistical Sciences, \textit{Università Cattolica di Sacro Cuore}
\\
\textbf{3} Department of Mathematics, \textit{Politecnico di Milano}

\section{Introduction}
Statistical inference, and in particular hypothesis testing, is central in the field of statistics. 
Functional data analysis (FDA) often deals with ``summary statistics'' such as identifying  mean curves\slash trajectories and principal modes of variation, including clustering and classification. 
Yet hypothesis testing is still a key part of FDA,  where it is often done in conjunction with \emph{functional regression}. Indeed, it is important to assess whether there is a significant effect of a covariate on the response, where covariates and response variables can be functions or scalars, depending on the setup.

In the common case of a functional response and scalar covariates, 
functional regression is usually modelled through a linear model on the form $y_i = A x_i + \text{noise}$, where $y_i$ is a curve belonging to a suitable function space e.g. $L^2[0,1]$,  $x_i$ are covariate(s), and $A \in L$ is a linear operator which ought to be estimated. 
In the simple case of one covariate, the null hypothesis would be asking whether $A = 0$, and more generally the null hypothesis would be asking if $A \in U$ for a given subspace $U \subset L$.

This is a hard question to ask for functional data and is not that frequently encountered. It is absent in some textbooks \citep{Ramsay, berlinetAgnan}  but has extensive treatment in \cite{horvathK2012}.
There are various approaches to this issue, which importantly also depend on the scope of the test. We will distinguish between  two kinds of tests, (i) \emph{global tests}: does covariate $x$ have  `influence' on curve $\theta$ in at least one part of the domain of $\theta$, and (ii) \emph{local tests} or \emph{domain selections}: if covariate $x$ has an influence, which part(s) of the domain of $\theta$ are significantly affected? 
 
\paragraph{Global tests}
Global tests have been studied  by various authors, for references see e.g. \cite{horvathK2012}. However, a crucial feature is that many of the available methods rely on (strong) parametric assumptions on the data distribution, such as Gaussianity\footnote{Note: This includes independence of PC scores (among other things)}, which may be valid asymptotically but can be problematic for the usually small sample sizes and infinite dimensions that are characteristic of functional data analysis.

Non-parametric approaches such as permutation tests are popular alternatives to parametric tests. Curves are (randomly) permuted wrt. likelihood-independent transformations, and some suitable test statistic (e.g. deviation from the mean) is evaluated for each permutation. 
However this is  computationally expensive, and permutation tests are only asymptotically exact in the presence of several covariates.

\paragraph{Local tests}
Local tests have not been studied to the same extent as global tests. In functional data analysis performing local inference carries several issues. The most important one is how to define, and how to control the probability of committing a type 1 error globally over the whole domain. 
One recent framework for doing local testing on functional data is the \emph{Interval-wise testing} (IWT) introduced by \cite{pv2017interval} and extended to linear models by \cite{abramowiczmox}. 
These procedures perform non-parametric inference based on permutation schemes and provide (asymptotical) control of the family-wise error rate (FWER) on each sub-interval of the domain. In detail, the probability of falsely selecting at least part of an interval where the null hypothesis is not violated is controlled.
Another procedure proposed in the literature, and focusing on the FWER control is the Fmax-procedure \citep{holmes1996, winkler2014}.  The Fmax-procedure is a  method that provides strong control of the FWER, and like the IWT-framework it is based on permutation tests. The procedure is multivariate in nature, but it can also be applied to functional data discretized on a fine grid. 

An alternative approach to local testing in functional or spatial data is to use discrete features of the observed functions such as local maxima or zero values. That is, instead of assessing a continuum of tests, one selects a finite number of data features for testing. 
Using discrete features has some challenge with respect to interpretation,  as one has to specify when two different observations can be considered instances of the same discrete feature, and there are no obvious definitions of domain selection.  \citep{cheng2017, schwartzman2011} present some interesting methods with this, where they also proof control of the false discovery rate (FDR). 

In this work we likewise focus on the control of the FDR. In particular, we focus on defining a continuum of statistical tests over the domain of functional data, and secondly we propose a generalization of the Benjamini-Hochberg procedure \citep{BH1995} for adjusting such tests in order to control the FDR.

\medskip 
The remainder of the paper is organised as follows: 
Section \ref{fdr-multi} describes false discovery rate in the multivariate case and reviews related work. 
Section \ref{fdr-fda-afsnit} presents the novel work of functional false discovery rate and the functional Benjamini-Hochberg procedure.
Two different simulation studies are presented in Section \ref{overafsnit-simulation}, 
and in Section \ref{klimadata-afsnit} we apply  the proposed procedure to a data set on climate change. Finally in Section \ref{diskussions-afsnit} we highlight and discuss important points of this article. Proof of the main theorems are provided in the appendix.

\section{False discovery rate for multivariate data} \label{fdr-multi}

\paragraph{Background}
\label{lit-review}

Multiple testing is a central topic in statistics. Observed within virtually every area of statistics, it is fundamental in many statistical applications and multiple testing is recognised as an important statistical issue within many sciences.

Many ways to deal with multiple testing have been proposed with various advantages and disadvantages -- one popular approach is to use the \emph{False Discovery Rate} (FDR) \citep{BH1995}. The false discovery rate is the expected proportion of false rejections (``discoveries'') among all among rejected hypotheses.
The procedures controlling the FDR are generally more powerful than the ones controlling the FWER.

FDR is often applied in cases when a single or comparatively few false positives is not considered a serious issue, as long as their rate among all discoveries can be controlled. 
In \citep{BH1995} the \emph{Benjamini-Hochberg \emph{(BH)} procedure} for controlling  FDR  is introduced. In the succeeding literature, a number of other procedures for controlling FDR have been proposed (see \cite{heesen2015} for a discussion). 
The paper  \cite{BY2001} is important, as it introduces a modification of the BH procedure that controls FDR without specifying any dependency assumptions. More importantly, they show that the original procedure introduced in \cite{BH1995}  controls FDR under a weaker assumption than independence, namely \emph{positive regression dependence} on the subset of true null hypotheses (PDRS). 

Other closely related quantities for assessing the errors within the paradigm of
multiple testing have been proposed such as the \emph{weighted false discovery rate} (WDFR) \citep{benjamini1997}, (see below) the \emph{positive false discovery rate} \citep{storey2003pfdr}, and the \emph{local false discovery rate} \citep{efron2001etal}.
 
However, in this paper we will only focus on the BH procedure and FDR which -- due to its simple interpretation and definition -- is still the most popular method for multiplicity correction. 
False Discovery Rate is only defined for finite numbers of hypotheses, and one must be careful when defining FDR on infinite sets of hypotheses.

\paragraph{False discovery rate} \label{fdr-def-afsnit}
Assume we are given a set of $m$ {null hypotheses}, $G_1, \dots, G_m$, each of which can either be true or false, and can either be accepted or rejected by a statistical test. 
Furthermore, let $w_1, \dots, w_m$ be strictly positive weights with $\sum w_i = 1$, which we assume are a priori known. 
This will be used in the case of weighted false discovery rate. 
These weights can e.g. be interpreted as how important the different tests are, where the "usual" false discovery rate corresponds to the case of equal weights.

\begin{defi}[False discovery rate, unweighted case]
The false discovery rate is defined as:
\begin{equation*}
    \E[Q] = \E\left[\frac{\#\{i : G_i \text{ is true but rejected}\}}
    {\#\{i : G_i \text{ is rejected}\}  } \right]
\end{equation*}
with $Q := 0$ whenever the denominator is zero. 
\end{defi}

\begin{defi}[False discovery rate, weighted case]
The false discovery rate in the weighted case is defined as:
\begin{equation*}
    \E[Q] = \E\left[ \frac{\sum\limits_{i : G_i \text{ is true but rejected}} w_i}
    {\sum\limits_{i : G_i \text{ is rejected}} w_i} \right] 
\end{equation*}
with $Q := 0$ whenever the denominator is zero. 
\end{defi}

\paragraph{The Benjamini-Hochberg procedure}
Let $\{p_{(i)}\}_{i=1}^m$ be the  $p$-values sorted in increasing order. Let $\{G_{(i)}\}_{i=1}^m$ be the corresponding ordering of the hypotheses and $\{w_{(i)}\}_{i=1}^m$ the corresponding ordering of weights.
The very popular and easily applicable \emph{Benjamini-Hochberg \emph{(BH)} procedure} for multiple comparison adjustment \citep{BH1995, benjamini1997} is defined as follows:

\begin{defi}[Benjamini-Hochberg procedure, unweighted case] \label{bh-def-unw}
Define
\begin{equation*}
k = \arg\max_i \left[ p_{(i)} \leq \frac{i}{m}
\dss \right]
\end{equation*}
The Benjamini-Hochberg procedure is: \emph{reject hypotheses $G_{(1)}, \dots, G_{(k)}$ corresponding to the $k$ smallest $p$-values and accept the rest.}
\end{defi}

\begin{defi}[Benjamini-Hochberg procedure, weighted case]  \label{bh-def-w}
Define
\begin{equation*}
k = \arg\max_i \left[ p_{(i)} \leq  \sum_{j: p_j \leq p(i)} w_j \dss \right]
\end{equation*}
The weighted Benjamini-Hochberg procedure is: \emph{reject hypotheses $G_{(1)}, \dots, G_{(k)}$ corresponding to the  $k$ smallest $p$-values and accept the rest.}
\end{defi}
Unlike most adjustment procedures introduced prior to this, the BH procedure is scalable: if data is duplicated such that one has twice the amount of hypotheses, the inference by using the BH procedure will still be the same. 

\paragraph{Adjusted p-values}

A concept often used in context of multiple testing is  \emph{adjusted }(or corrected)\emph{ p-values}. Informally, the adjusted p-values for a multiple testing procedure are defined as corrections $\{\tilde{p}_i\}$ of the original $p$-values such that a null hypothesis $G_i$ can be rejected at level $\dss$ if $\tilde{p}_i \leq \dss$.

The adjusted p-values for the unweighted Benjamini-Hochberg procedure are
\begin{equation*}
	\tilde{p}_{(i)} = \min(1, \tfrac{m}{i} p_{(i)}, \dots, \tfrac{m}{m-1} p_{(m-1)},  p_{(m)} ) 
\end{equation*}
where $p_{(\cdot)}$ and $\tilde{p}_{(\cdot)}$ are the order statistics of $p$ and $\tilde{p}$, respectively. By construction, the ordering is the same.

We can likewise define adjusted p-values for the weighted Benjamini-Hochberg procedure:
%
\begin{equation*}
\tilde{p}_{(i)} = \min \left\{1,  \frac{p_{(i)}}{\sum_{j: p_j \leq p_{(i)}} w_j}, 
\frac{p_{(i+1)}}{\sum_{j: p_j \leq p_{(i+1)}} w_j}, 
\dots, 
\frac{p_{(m-1)}}{\sum_{j: p_j \leq p_{(m-1)}} w_j}, 
 p_{(m)}  \right\}, \quad i = 1, \dots m 
\end{equation*}

\paragraph{PRDS and control of false discovery rate}
Benjamini and Hochberg showed in their seminal paper \cite{BH1995} that if test statistics for different hypotheses are independent, then the false discovery rate is controlled by $\frac{m_0}{m}\alpha$ where $m_0$ is the total number of correct null hypotheses. The independence assumption was later relaxed by \cite{BY2001} to \emph{positive regression dependency on one} (PRDS).

Below we define the  PRDS property and extend it to the infinite-dimensional case, which will be needed later.

\begin{defi}[Positive regression dependency on one (PRDS)] \label{pdrs-def}
Let '$\leq$' be the usual ordering on $\R^l$.
	An \emph{increasing set} $D \subseteq \R^l$ is a set satisfying $x \in D \wedge y \geq x \Rightarrow y \in D$.
	
	A random variable $\mathbf{X}$ on $\R^l$ is said to be \emph{PRDS on  $I_0$},  where $I_0$ is a subset of $\{1, \dots, l\}$, if it for any increasing set $D$ and $i \in I_0$ holds that 
\begin{equation} \label{eq-pdrs-def}
	x \leq y \Rightarrow P(\mathbf{X} \in D | X_i = x) \leq P(\mathbf{X} \in D | X_i = y)
\end{equation}

Let  $\mathbf{Z}$ be an  infinite-dimensional random variable, where instances of $\mathbf{Z}$ are functions $T \pil \R$. We say that $\mathbf{Z}$ is PRDS on $U \subseteq T$ if all finite-dimensional distributions of $\mathbf{Z}$ are PRDS. That is, for all finite subsets $I = \{i_1, \dots , i_l \} \subseteq T$, it holds that
 $Z(i_1), \dots, Z(i_l)$ is PDRS on $I \cap U$.
\end{defi}
We refer to \cite{BY2001} for a discussion on the PRDS property and how it relates to other types of dependency.

\begin{thm} \label{pdrs-thm} 
	Given a set of hypotheses $\{ H_1, \dots, H_m \}$ and corresponding p-values $(p_1, \dots, p_m) $,
	let $I_0 = \{i_1, \dots, i_{m_0} \} \subseteq \{1, \dots, m\}$ be the index set corresponding to true null hypotheses $\{H_{i_1}, \dots , H_{i_{m_0}} \}$.

	If the joint distribution of the p-values $(p_1, \dots, p_m) $ is PDRS on  $I_0$, the BH procedure controls the FDR at level $\frac{m_0}{m} \dss$ in the sense that
\[ \E [Q] \leq \frac{m_0}{m} \dss \leq \dss
\]
where $Q$ is the proportion of false discoveries.
\end{thm}

\begin{proof}
	See \cite[Theorem 1.2]{BY2001}
\end{proof}

\section{False discovery rate for functional data} \label{fdr-fda-afsnit}

In this section we define the false discovery rate for functional data and propose a functional extension of the Benjamini-Hochberg procedure. These are the functional versions of the classical false discovery rate and the BH procedure used in the multivariate cases.

\subsection{Definition of functional false discovery rate}

Our definition of false discovery rate is related to that of \cite{sun2015etal}, although \cite{sun2015etal} uses a seemingly arbitrary lower bound for the measure of rejection region, not present in other papers on FDR, and avoids references to p-values which are natural 
in the context of multiple testing. 

Having defined functional false discovery rate and introduced the functional Benjamini-Hochberg procedure, we proceed to define the adjusted p-value function, an important tool for practical applications, in Section  \ref{sec-justeret-p}. Section \ref{sec-finite-approx} contains a discussion and results about  finite approximations of the fBH procedure, which is easily adopted into an algorithm. Finally we prove control of the false discovery rate under regularity conditions in Section \ref{fdr-kontrol-thm}.

For the remainder of this section we assume that we have $N$ functional samples, $y_1, \dots, y_N$, with $y_i :\mathbf{T} \pil \R$, where $ \mathbf{T} \subset R^d$ is an open and bounded subset of $\R^d$.

Suppose that for each point $t \in \mathbf{T}$, we have a null hypothesis $H^0_t$ together with an alternative hypothesis $H^A_t$, that we are interesting in testing. 
Furthermore, suppose that by pointwise application of some statistical test we obtain p-values $p(t)$ for every $t$ with the property that $(H^0_t$ true$) \implies p(t) \sim U(0,1)$, where $U(0,1)$ is the uniform distribution on $(0,1)$.   
The $p$-values together make up a function $p: \bt \pil [0,1]$, the \emph{unadjusted $p$-value function} \citep{pv2017interval}.

Let $U$ be the set of the domain where the null hypothesis is true, ie. $U = \{t \in \bt: H_0^t$ is true $\}$. %
Let $\nu$ be a bounded measure on $\bt$ that is absolutely continuous wrt. the Lebesgue measure, which we denote by $\mu$.  By absolute continuity we have $\nu = f \cdot \mu$ for some measurable function $f: \bt \pil [0,\infty)$. The function $f$ can be interpreted as a \emph{weight function} assigning more weight to some regions of $\bt$ than others, and is the functional counterpart of the weights used in section \ref{fdr-def-afsnit}.

Given $U$ and an instance of $p(t)$, 
let $V = \{t: H^0_t$ is true and $H^0_t$ is rejected$\}$ be the region where the null hypothesis is wrongly rejected, and let $S = \{t: H^0_t$ is false and $H^0_t$ is rejected$\}$ be the region where the null hypothesis is correctly rejected. 
The set $V$ corresponds to committing type I errors, and in a given research situation, it is desired that $V$ is as small as possible and $S$ is as large as possible.
Specifically, the functional false discovery rate can be defined as follows.
\begin{defi}[Functional false discovery rate]
Define the  \emph{functional false discovery rate} (fFDR) as 
\begin{equation}
\fFDR = \E[Q] = \E\left[\frac{\nu(V)}{\nu(V \cup S)} 1_{\nu(V \cup S) > 0} \right] \label{funk-fdr-def}
\end{equation}
where $Q$ is the \emph{proportion of false discoveries}.     
\end{defi}

\begin{remark}[False discovery rate for other manifolds] \label{mangfoldighedsremark} 
In this paper we define false discovery rate for functional data defined on open subsets of $\R^d$. However, many smooth manifolds can be diffeomorphically mapped into open and bounded subsets of $\R^d$. The mapping gives naturally rise to a measure on this set, which can be used as measure for functional false discovery rate.
\end{remark}

\subsection{The functional Benjamini-Hochberg procedure: the adjusted threshold}

Analogous to the multivariate case, we can define the fBH procedure, that is the Benjamini-Hochberg procedure for functional data. The counts and sums used in Defintions \ref{bh-def-unw} and \ref{bh-def-w} are replaced by the measure $\nu$, that incorporates any weighting of the domain $\bt$. 

The main theoretical result of this article is that the fBH procedure can be approximated by the multivariate BH procedure on the pointwise evaluations of functional data, and that it controls the fFDR by $\alpha \nu(U)/\nu(\bt)$ under regularity conditions. 

\begin{defi}[Functional Benjamini-Hochberg procedure]
	\label{bh-def-unweighted}
	Let $\tss \in (0,1)$ be a desired significance level for the tests. 
The functional Benjamini-Hochberg (fBH) procedure is: 

\emph{Reject 	hypotheses $H^0_t$ that satisfy}
\[
p(t) \leq \alpha^* \quad \text{where}\quad {\alpha}^* = \arg \max_r  \frac{\nu(\{s: p(s) \leq r \})}{\nu(\bt)} \geq \tss^{-1} r
\]
We will refer to $\alpha^*$ as the \emph{adjusted threshold} of the procedure, and the function $a(r) =  \nu(\{s: p(s) \leq r \})$ as the \emph{cumulated p-value function}.
\end{defi}

Two examples of the functional BH procedure are shown in Figure \ref{fig-bh-eksempel}.

\begin{figure}[!]
	\centering
	\includegraphics[width=0.49\textwidth]{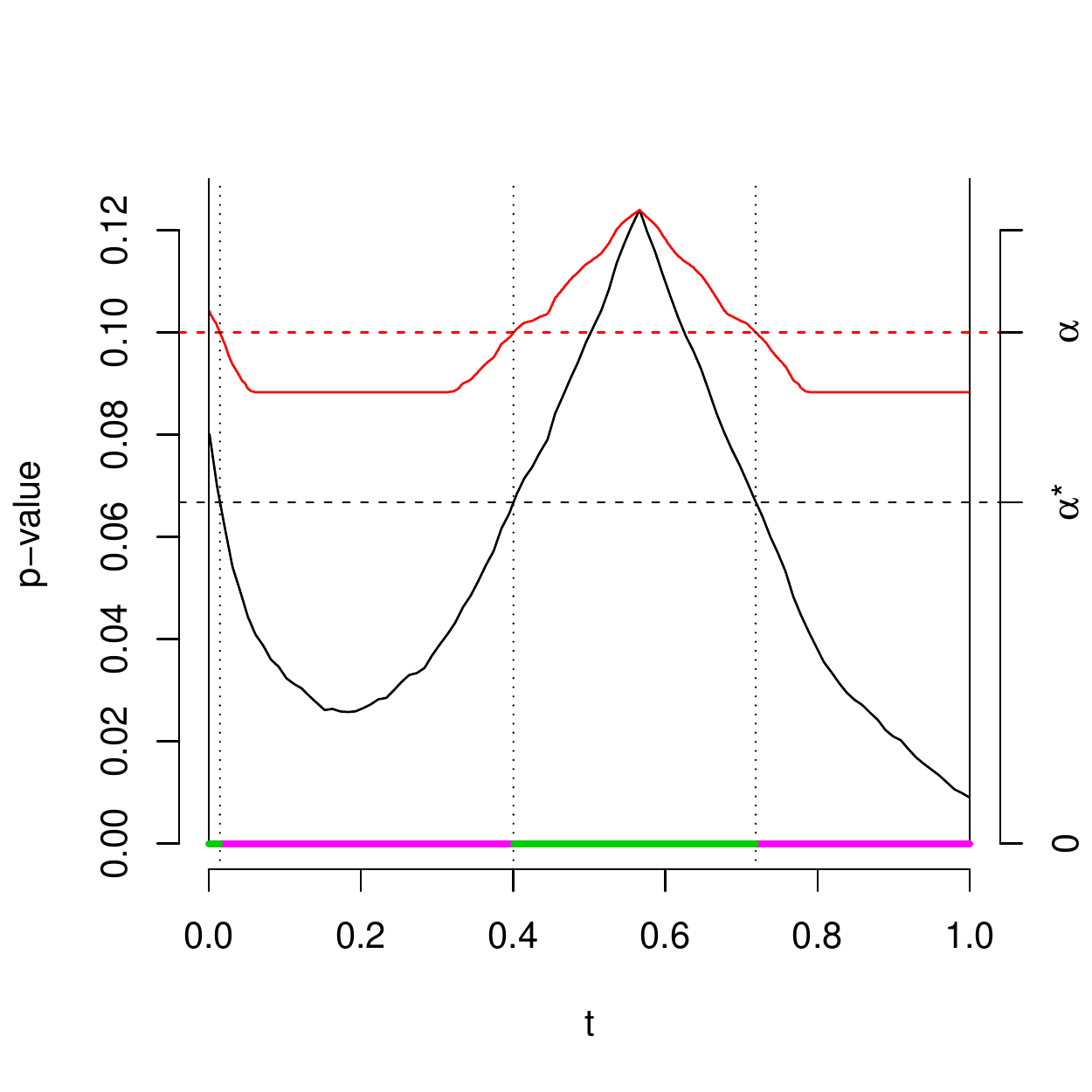}
	\includegraphics[width=0.49\textwidth]{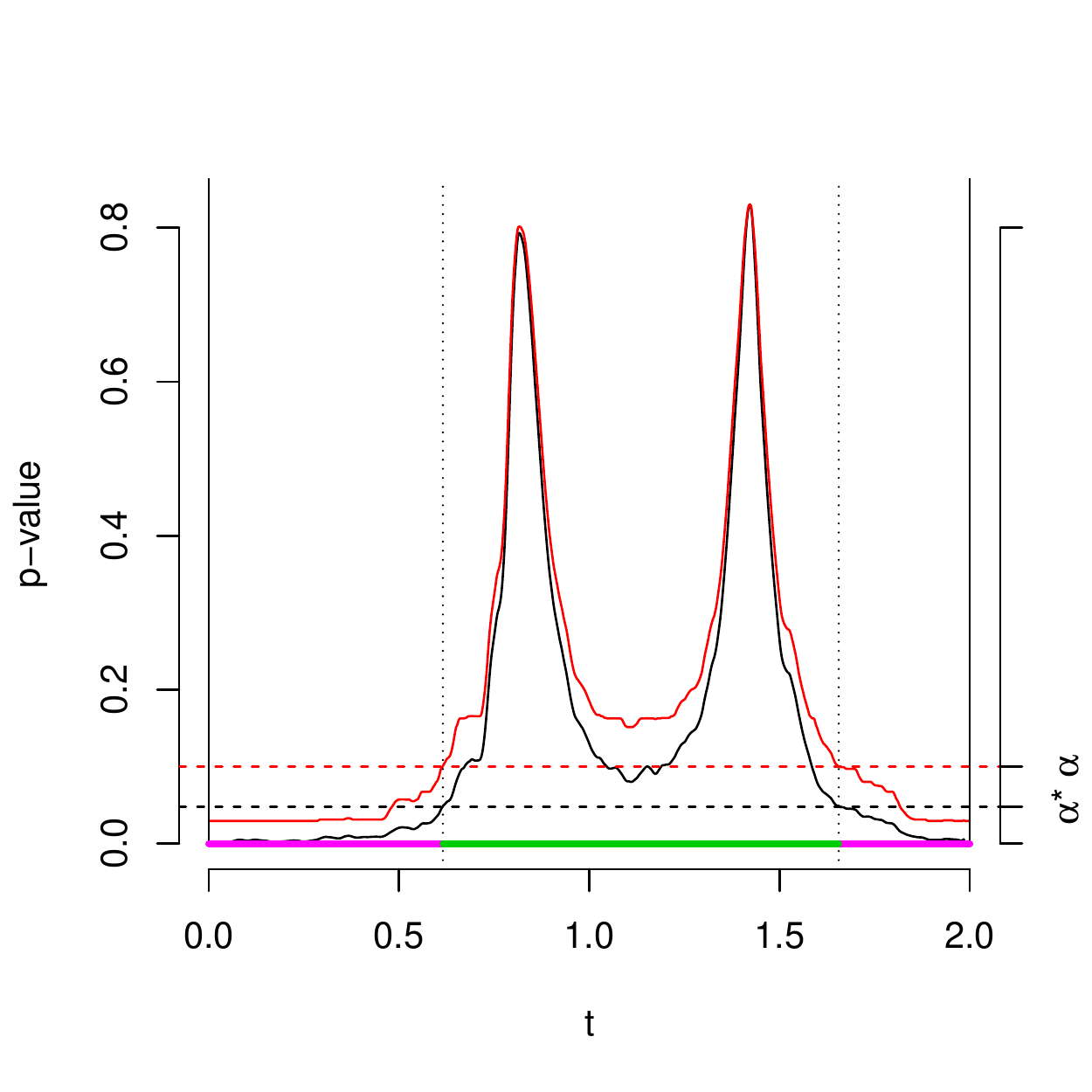} \\
	\includegraphics[width=0.49\textwidth]{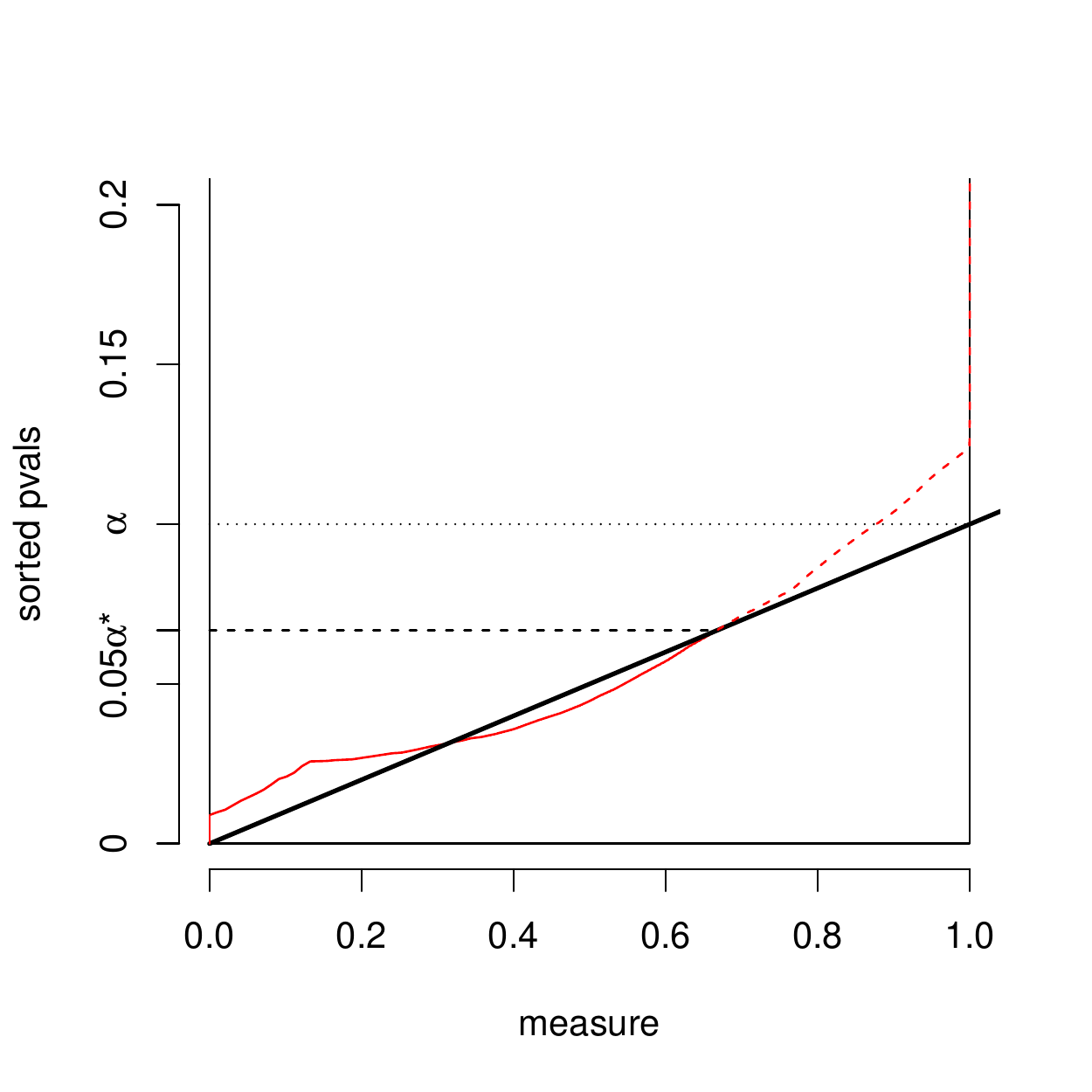}
	\includegraphics[width=0.49\textwidth]{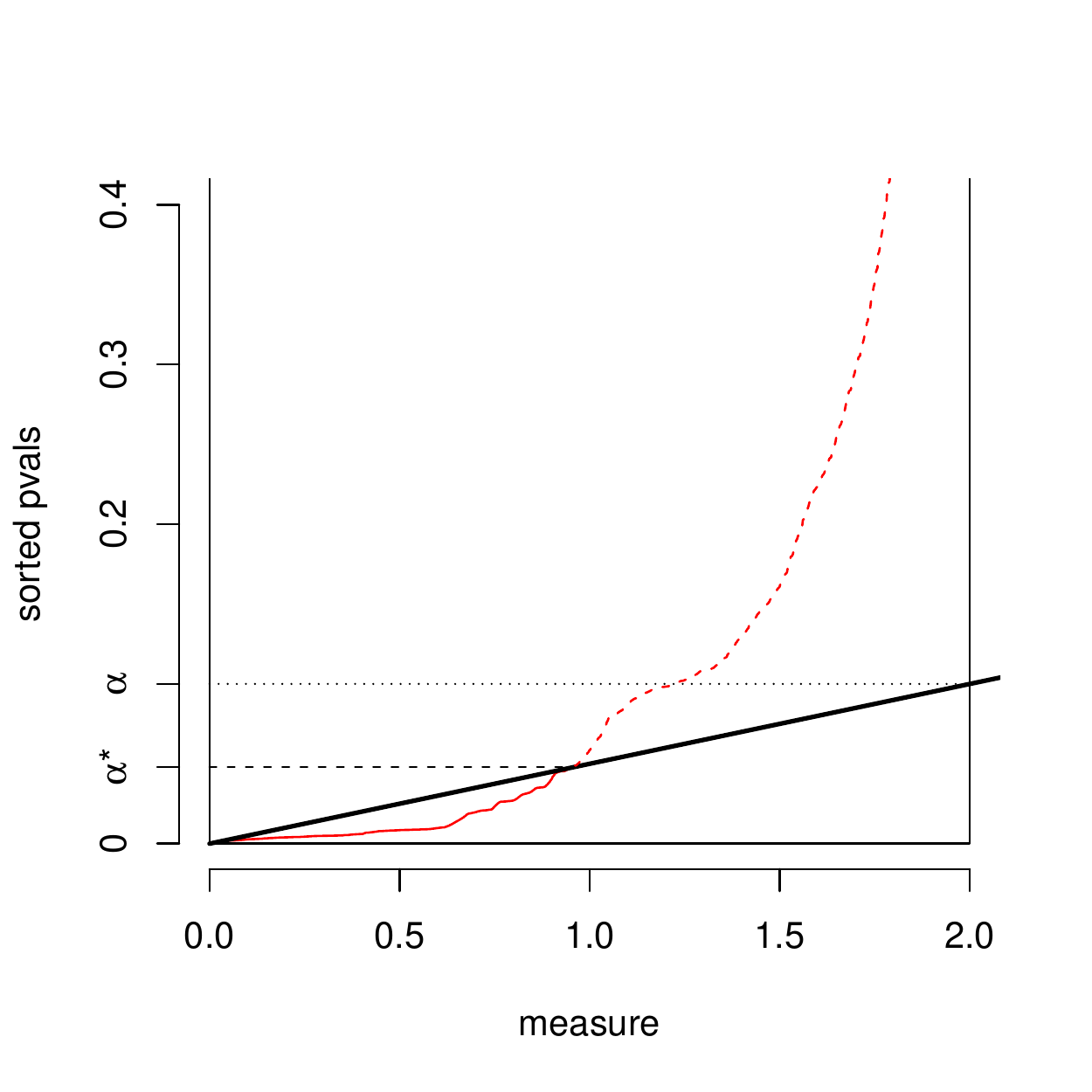} 
	\caption{Two illustrations of the functional Benjamini-Hochberg procedure and adjusted p-value functions with $\alpha = 0.10$. Upper plots: black curves are original p-values; red curves are adjusted p-values; green and purple indicate regions where $H^0_t$ is accepted/rejected after adjustment, respectively. \\
	Lower plots: The red lines  denote the cumulated p-value functions; the thick lines have slope $\alpha$.
  Null hypotheses corresponding to the solid red lines are rejected, while those above are accepted.} \label{fig-bh-eksempel}
\end{figure}

\subsection{The functional Benjamini-Hochberg procedure: the adjusted p-value function} \label{sec-justeret-p}
In the application of FDR-controlling procedures, it is often of interest to have the possibility of changing the threshold. 
For the functional case, this amounts to asking for which $\alpha$ that $ H_0^t$ is rejected after adjustment by the  fBH procedure. This naturally leads to the notion of the  \emph{fFDR-adjusted p-value function}:
\begin{equation*}
\tilde{p}(t) = 
\min_{ s \geq p(t)} \left\{1,  \frac{\nu(\bt)s}{\nu(r: p(r) \leq s)} \right\}
, \quad t \in \bt
\end{equation*}
This is analogous to adjusted p-values, which applies in the discrete case, and  plays a similar role, ie. if $\tilde{p}(t) \leq \alpha$ this means that $H^0_t$ will be rejected when the (weighted) fBH procedure is applied with threshold level $\alpha$. 

Adjusted $p$-values allow us to quantify significance after adjustment and to simultaneously compute rejection areas for all values of $\alpha$. By \Cref{adjust-equiv} below, adjustment using either the adjusted threshold or adjusted p-value function gives rise to the same inference.

\begin{prop} 
\label{adjust-equiv}
Let $\tss$ be a given threshold, and define  the adjusted threshold $\alpha^*$ according to \Cref{bh-def-unweighted}. Then adjustment using the adjusted threshold or the adjusted p-value function gives the identical results, ie.
\[ p(t) \leq \alpha^* \Longleftrightarrow \tilde{p}(t) \leq \tss, \quad t \in \bt
\]
\end{prop}
\begin{proof}
Using the definitions of the adjusted threshold and the adjusted p-value function, respectively, the result is true since
\begin{multline*}
p(t) \leq \alpha^* \eb \exists s \geq p(t):  \frac{\nu(\{r: p(r) \leq s \})}{\nu(\bt)} \geq \tss^{-1}s \Longleftrightarrow
\\
 \exists s \geq p(t):  \frac{\nu(\bt) s}{\nu(\{r: p(r) \leq s \})} \leq \tss \eb \tilde{p}(t) \leq \tss
\end{multline*}
\end{proof}

\subsection{The functional Benjamini-Hochberg procedure: Finite approximation and computational cost} \label{sec-finite-approx}

In Section \ref{fdr-kontrol-thm} (see Theorem \ref{thm-control}), we will prove that the fBH procedure controls the fFDR. Before doing so, we start defining a finite approximation of the fBH procedure, that is proven to converge to the theoretical one. 
Indeed, a crucial issue regarding the  fBH procedure is that it makes use of an infinite amount of hypotheses and p-values, which in principle is computationally unattainable. 
We present here a simple algorithm for approximating the fBH procedure. In Propositions \ref{hoved-thm} and \ref{general-thm} we show that under regularity conditions,  this algorithm approaches the fBH procedure in the limit.

We present the algorithm and proposition in two versions: a notationally simpler version with equal weights, which is a special case of the more general, weighted version.

\subsubsection{Unweighted case} In many applications there is no \textit{a priori} reason for assigning different weights to different parts of the domain, $\bt$. Thus we will assume that $\nu = \mu$ corresponding to the weight function $f$ being constant.  

\begin{algorithm}[Unweighted case]
Let $S$ be a dense, uniform grid in $\bt$. The fBH procedure can be approximated by evaluation of the `usual' BH procedure:\\
Apply the BH procedure to $\{p(t) : t \in S\}$.
\end{algorithm}
Obviously, not all grids are suitable, and what a 'dense uniform grid' constitutes, is stated in Proposition \ref{hoved-thm}. 
The typical choice of a grid would be a lattice, and by decreasing the mesh of the lattice, the approximation becomes better and better.

\begin{prop}\label{hoved-thm}
	
Let $\{S_k\}_{k=1}^\infty, S_1 \subset S_2 \subset \dots$ be a dense, uniform grid in $\bt$ in sense that $S_k$ uniformly approximates all level sets of $p$ and $p|_U$ with probability one:
\begin{equation}
P \left[ \lim_{k \pil \infty} \sup_r \frac{\#( S_k \cap \{s: p(s) \leq r \})} {\# S_k} -  \mu\{s: p(s) \leq r \} \pil 0 \right] = 1
    \label{assump1-leb}
\end{equation} and
\begin{equation}
P \left[ \lim_{k \pil \infty} \sup_r \frac{\#( S_k \cap \{s: p(s) \leq r \} \cap U)} {\# S_k} -  \mu(\{s: p(s) \leq r \} \cap U) \pil 0 \right] = 1 \label{assump2-leb}
\end{equation}
Furthermore, assume that the assumptions about the $p$-value function below hold true with probability one:
\begin{enumerate}
	\item[(a1)] All level sets of $p$ have zero measure, 
	\[ \mu\{s: p(s) = t \} = 0 \quad \forall t \in \bt
	\]
	\item[(a2)] $\alpha^* \in (0, \tss] \Rightarrow$: for any open neighbourhood $O$ around $\alpha^*$ there exists $s_1, s_2 \in O$ s.t. $ a(s_1) > \alpha^{-1} s_1, a(s_2) < \alpha^{-1} s_2$, where $a$ is the cumulated p-value function (Definition \ref{bh-def-unweighted}).
	\item[(a3)]  $[\alpha^* = 0] \Rightarrow \min p(t) > 0$. 
\end{enumerate}

Now define the k'th step false discovery proportion $Q_k$  by applying the (usual) BH procedure at level $\tss$ to $p$ evaluated in $S_k$.

Mathematically, this can be defined by  
\begin{equation}
    Q_k = \frac{\#\{t \in S_k : p(t) \leq b_k\} \cap U}{\#\{t \in S_k : p(t) \leq b_k\}}, \quad 
    b_k = \arg\max_r \frac{\# \{s \in S_k : p(s) \leq r \}}{\#S_k} \geq \alpha^{-1} r \label{qk-bk-formel}
\end{equation}
Then $Q_k$  behaves asymptotically as the functional false discovery proportion:
\begin{equation*}
    \lim_{k \pil \infty} Q_k = Q. 
\end{equation*}
where $Q$ is defined as in \eqref{funk-fdr-def}.
Furthermore, if $p$ is PDRS wrt. the set of true null hypotheses with probability one, then the false discovery rate $\E[Q_k]$ is controlled by $\tss \mu(U)/\mu(\bt)$:
\begin{equation*}
     \limsup_{k \pil \infty} \E[Q_k] \leq \tss \frac{\mu(U)}{\mu(\bt)} \leq \tss 
\end{equation*}

\end{prop}
\begin{proof}
See appendix. Also a special case of \Cref{general-thm}.
\end{proof}
From the proof of the theorem, we have the following important corollary which states that as the grid $S_k$ becomes tighter and tighter, hypotheses are eventually rejected or accepted:
\begin{korr} \label{konvergens-htk}
For $t \in \cup_{m=1}^\infty S_m$ and $k\geq 1$, 
define $H_{t,k} =  (t \in S_k) \wedge (p(t) \leq b_k)$ where 
$b_k$ is is given by \eqref{qk-bk-formel}.
That is, $H_{t,k}$ is true if the adjusted threshold at step $k$ is larger than $p(t)$. 
Assume $p(t) \neq \tss$. Eventually, $H_{t,k}$ is either rejected or accepted.
\end{korr}
\begin{proof}
Proposition \ref{htk-konv}
\end{proof}

\subsubsection{Weighted case} 
Allowing the weighting function, $f$, to vary allows for more general measures, such as those  which arise from \Cref{mangfoldighedsremark}. However, the notation is more tedious, so we have separated unweighted case as a special case. 

Let $\nu = f \cdot \mu$, and assume that $f$ is a bounded and strictly positive density function $\bt \pil \R$. 
\begin{algorithm}[Weighted case]
Let $S$ be a dense, uniform grid in $\bt$. The fBH procedure can be approximated by evaluation of the 'usual' BH procedure: \\
Apply the BH procedure to $\{p(t) : t \in S\}$ with weights $(f(t): t \in S)$ where the weights have been normalised to one.
\end{algorithm}

Under almost similar assumptions to the unweighted case, we are able to control the false discovery rate at level $\alpha$:
\begin{prop} \label{general-thm}
Let $\{S_k\}_{k=1}^\infty, S_1 \subset S_2 \subset \dots$ be a dense, uniform grid in $\bt$ in sense that $S_k$ weighted by $f$ uniformly approximates all level sets of $p$ and $p|_U$ with probability one:
\begin{equation}
P \left[ \lim_{k \pil \infty} \sup_r \frac{\sum_{i \in S_k \cap \{s: p(s) \leq r \}} f(i) } {\# S_k} -  \int_{\{s: p(s) \leq t \}} f(x) \de x  \pil 0 \right] = 1
    \label{assump1-gen}
\end{equation} and
\begin{equation}
P \left[ \lim_{k \pil \infty} \sup_r \frac{\sum_{i \in S_k \cap \{s: p(s) \leq r \} \cap U} f(i) } {\# S_k} -  \int_{\{s: p(s) \leq t \} \cap U} f(x) \de x  \pil 0 \right] = 1
\label{assump2-gen}
\end{equation}

Furthermore, assume that the assumptions about the $p$-value function below hold true with probability one:
\begin{enumerate}
	\item[(a1)] All level sets of $p$ have zero measure, 
	\[ \nu\{s: p(s) = t \} = 0 \quad \forall t \in \bt
	\]
	\item[(a2)] $\alpha^* \in (0, \tss] \Rightarrow$: for any open neighbourhood $O$ around $\alpha^*$ there exists $s_1, s_2 \in O$ s.t. $ a(s_1) > \alpha^{-1}s_1, a(s_2) < \alpha^{-1}s_2$, where $a$ is the cumulated p-value function (Definition \ref{bh-def-unweighted}).
	\item[(a3)] $[\alpha^* = 0] \Rightarrow \min p(t) > 0$.
\end{enumerate}

Now define the $k$th step false discovery proportion $Q_k$  by applying the (usual) BH procedure at level $\tss$ to $p$ evaluated in $S_k$, weighted by $f$ evaluated in $S_k$.

Mathematically, this can be defined by  
\begin{equation}
    Q_k = \frac{\sum_{t \in S_k \cap U : p(t) \leq b_k} f(t)}
    {\sum_{t \in S_k : p(t) \leq b_k} f(t)}, \quad 
    b_k = \arg\max_r \frac{\sum_{\{s \in S_k : p(s) \leq r \}} f(s)}{\sum_{\{s \in S_k\}} f(s)} \geq \alpha^{-1}r
 \label{qk-bk-formel-w}
\end{equation}

Then $Q_k$  behaves asymptotically as the functional false discovery proportion:
\begin{equation*}
    \lim_{k \pil \infty} Q_k = Q. 
\end{equation*}
where $Q$ is defined as in \eqref{funk-fdr-def}.
Furthermore, if $p$ is PDRS wrt. the set of true null hypotheses with probability one, then the false discovery rate $\E[Q_k]$ is controlled by $\tss \nu(U)/ \nu(\bt)$:
\begin{equation*}
     \limsup_{k \pil \infty} \E[Q_k] \leq \tss \frac{\nu(U)}{\nu(\bt)} \leq \alpha 
\end{equation*}

\end{prop}

\begin{proof}
	See appendix.
\end{proof}
Analogous to \Cref{konvergens-htk} we have the following corollary which states that as the grid becomes tighter and tighter, hypotheses are eventually rejected or accepted:
\begin{korr} \label{konvergens-htk-gen}
For $t \in \cup_{m=1}^\infty S_m$ and $k\geq 1$, 
define $H_{t,k} =  (t \in S_k) \wedge (p(t) \leq b_k)$ where $b_k$ is given by \eqref{qk-bk-formel-w}.
That is, $H_t$ is true if the adjusted threshold at step $k$ is larger than $p(t)$. 
Assume $p(t) \neq \alpha$. Eventually, $H_{t,k}$ is either rejected or accepted.
\end{korr} 

The results in this section are important as they outline how to do inference using the fBH procedure. The implementation is very easy and requires no additional tools besides an implementation of the multivariate BH procedure.


\paragraph{Computational cost}
The fBH procedure adds $O(n \log n)$  computational cost to the calculation of adjusted p-values, where $n = \#S_k$ is the number of approximation points,. This extra computational cost is  due to the sorting of $p$-values, which has $O(n \log n)$ computational cost. As a default, calculation of pointwise $p$-values has complexity $O(n)$, thus the total computational cost is $O(n \log n)$.

However, sorting algorithms on modern computer are very fast, whereas calculating pointwise p-values can be comparatively slow, in particular if permutation tests are used, such as in Section \ref{sim-1d-afsnit}. Thus, the computational cost from fBH adjustment procedure is expected to be negligible  in practice. 
This was also the case for the simulation studies and the case study presented in Sections \ref{overafsnit-simulation} and \ref{klimadata-afsnit}.

\subsection{Control of false discovery rate for functional data} \label{fdr-kontrol-thm}

With \Cref{general-thm} in place, it is easy to prove that the functional Benjamini-Hochberg procedure controls FDR under PRDS and regularity assumptions. 

\begin{thm}\label{thm-control}
Assume that there exists a uniform grid $\{S_k\}_{k=1}^\infty, S_1 \subset S_2 \subset \dots$ dense in $\bt$, such that with probability one assumptions \eqref{assump1-gen} and \eqref{assump2-gen} are met. 
Furthermore, assume that $p(t)$ is PRDS wrt. the set of true null hypotheses with probability one, and that the assumptions about $p$-value function below hold true with probability one:
\begin{enumerate}
	\item[(a1)] All level sets of $p$ have zero measure, 
	\[ \nu\{s: p(s) = t \} = 0 \quad \forall t \in \bt
	\]
	\item[(a2)] $\alpha^* \in (0,\tss] \Rightarrow$: for any open neighbourhood $O$ around $\alpha^*$ there exists $s_1, s_2 \in O$ s.t. $ a(s_1) > \alpha^{-1}s_1, a(s_2) < \alpha^{-1}s_2$, where $a$ is the cumulated p-value function (Definition \ref{bh-def-unweighted}).
	\item[(a3)] $[\alpha^* = 0] \Rightarrow \min p(t) > 0$.
\end{enumerate}

Then the functional BH procedure controls FDR at level $\tss \nu(U) / \nu(\bt)$, ie. $\E[Q] \leq \tss \nu(U) / \nu(\bt) \leq \tss$, when applying the functional BH procedure at level $\tss$.

\end{thm}

Note that the assumptions \eqref{assump1-gen} and \eqref{assump2-gen} are much simplified in the equal-weight case where $f \equiv 1$.

\begin{proof} For the ease of presentation, we will only show the proof in the equal-weight case, ie. when $f$ is constant (and thus can be disregarded).

From  \Cref{general-thm} and the notation of that proposition, it shows that $Q_k \pil Q$ almost surely. 

Since the PDRS assumption is fullfilled for all finite sets, we can apply  Benjamini and Yekuteli's original proposition to $S_k$ (Theorem \ref{pdrs-thm}), from which it holds that
$E[Q_k] \leq \frac{\#(S_k \cap U)}{\#S_k} \alpha$. 

As $0 \leq Q_k \leq 1$ for all $k$, it is now a simple application of the dominated convergence theorem to show that $\E[Q] \leq \alpha \nu(U)/\nu(\bt)$:
\begin{align*}
	\E[Q] &= \E[\lim_{k \pil \infty} Q_k] =  \lim_{k \pil \infty} \E[Q_k]  \\ 
	&= \limsup_{k \pil \infty} \E[Q_k]  \leq  \limsup_{k \pil \infty} \frac{\#(S_k \cap U)}{\#S_k} \tss =  \tss \frac{\nu(U)}{\nu(\bt)} \leq \tss,
\end{align*}
where the last equation follows from assumption \eqref{assump2-gen} (cf.\ its unweigthed form \eqref{assump2-leb}).
\end{proof}


As remarked in \cite{BY2001} the PDRS assumption is sufficient but not necessary, and \eqref{eq-pdrs-def} needs only to be true for certain sets defined in relation to the order statistics of $p$.
The details are quite technical, and we refer to \cite[Remark 4.2]{BY2001} and the general discussion of that paper.

\paragraph{Sufficient criteria for the one-dimensional case}

The assumptions (a1)-(a3) and equations \eqref{assump1-gen} and \eqref{assump2-gen} are consequences of the more simple criteria in the one-dimensional case, which must be true with probability one: 
\begin{enumerate}
	\item[(d1)] $p$ is continuous. 
	\item[(d2)] There exists a maximal number of crossings, $N_C$,  i.e.: 
	\[ \# \{s \in [0,D] |p(s) = t \} \leq N_C \quad \forall  t \in (0,1] 
	\]
	\item[(d3)]  $U$ is a finite union of disjoint intervals. 
\end{enumerate}
These criteria will generally be true for smooth curves as we typically use to model functional data. 


\section{Simulations} \label{overafsnit-simulation}

In this section two different simulation studies are performed, which differ both in scope and setting -- in the first simulation study, a functional-on-scalar regression is performed using permutation tests, while in the second simulation we test for mean equal to zero using one-sided t-tests.

The first simulation study has a more theoretical flavour and compares our proposed method to the Fmax method. The second study is intended to simulate a more realistic scenario with a number of disjoint peaks, and 
the performance of our method for various levels of $\alpha$ is analysed.

\subsection{1D-simulation} \label{sim-1d-afsnit}

\paragraph{Description of simulation}
In this section we wanted to numerically assess the performance of our fBH procedure, and to compare it with the Fmax-procedure \citep{holmes1996, winkler2014}.  The Fmax-procedure is a  method that provides strong control of the family-wise error rate in a multivariate high-dimensional setting.  It can be extended to functional data by applying it to the discrete point-wise evaluations of the functional data.

We simulated functional data according to the following functional-on-scalar linear model:
\begin{equation*}
    y_i(t) = \beta(t) x_{i} + \varepsilon_i(t) \quad i=1,\ldots,n, \quad t \in (0,1)
\end{equation*} 
where $n=10$, $x_{i} = \frac{i-1}{n-1}$, and $\beta(t) = d \cdot f(t)$, with $d$ ranging from 0 to 5. 
We modelled the function $f(t)$ with a cubic B-spline expansion with 40 basis functions and equally-spaced knots. The first $h$ coefficients of the expansion were set to one, and the last $40-h$ coefficients were set to zero. 
The resulting function assumed the value 1 in the first part of the domain,  0 in the second part of the domain, with a smooth transition.
We explored three  values of the parameter $h$: $h \in \{10, 20, 30\}$.

The error functions $\varepsilon_i( t) $ were obtained by simulating the coefficients of the same cubic B-spline expansion. The 40 coefficients were sampled independently from a standard normal distribution. 
The three panels in the first column of Figure \ref{fig:simulation_1D}  show an instance of the simulated functional data with $d=5$, and $h=10, 20, 30$, respectively. The functional data are coloured in grayscale that is proportional to the value of $x_i$. The study is much similar to the simulation study presented in \cite{abramowiczmox}, but here it is of interest to vary the domain where the null hypothesis is true. 

The fBH and Fmax procedures are applied to test the following hypotheses:
\begin{equation*}
    H_t^0: \beta(t)=0; \quad H_t^1: \beta(t) \neq 0.
\end{equation*}
The unadjusted $p$-value at point $t$ was computed with a permutation test based on the \emph{Freedman and Lane} method \citep{freedman1983nonstochastic}. 

With $d\in \{0,1,\ldots,5 \}$ and $h \in \{10, 20, 30\}$, we had 18 scenarios in total, but only 16 different ones, as the scenarios were identical for $d = 0$. 
\paragraph{Simulation results}

We measured the performances of the two methods by evaluating the FWER, FDR, false positivity rate (i.e., the measure  of the incorrectly rejected part of the domain over the measure of the domain where the null hypothesis is true) and 
sensitivity (i.e., the measure  of the correctly rejected  part of the domain relative to the total measure of the domain where the null hypothesis is false).
We performed the tests at nominal level $\alpha=0.05$.

Figure \ref{fig:simulation_1D} reports the results of the simulation obtained averaging over 1000 instances. Each row of the figure report the results with a different value of $h$. Each panel on columns 2-5 reports one of the measures discussed before as a function of the parameter $d$ for the unadjusted $p$-value (black), the fBH procedure (dark grey) and the Fmax procedure (light grey).

As expected by theory, the unadjusted $p$-value function does not control the FWER nor the FDR. It controls instead the false positive rate. 
The Fmax method controls the FWER, and by consequence, also the FDR. Finally, the fBH method controls the FWER only weakly (i.e., when $d=0$ and by consequence $H_t^0$ is true for all $t$). It controls instead the FDR in all scenarios. 
We notice a trade-off between the type of control and the sensitivity of the methods. The Fmax, being provided with a stronger control, is also the less sensible to deviations from the null hypothesis. Instead, the fBH is provided with a less strong control, but it is more sensible.

\begin{figure}
    \centering
    \includegraphics[width = \textwidth]{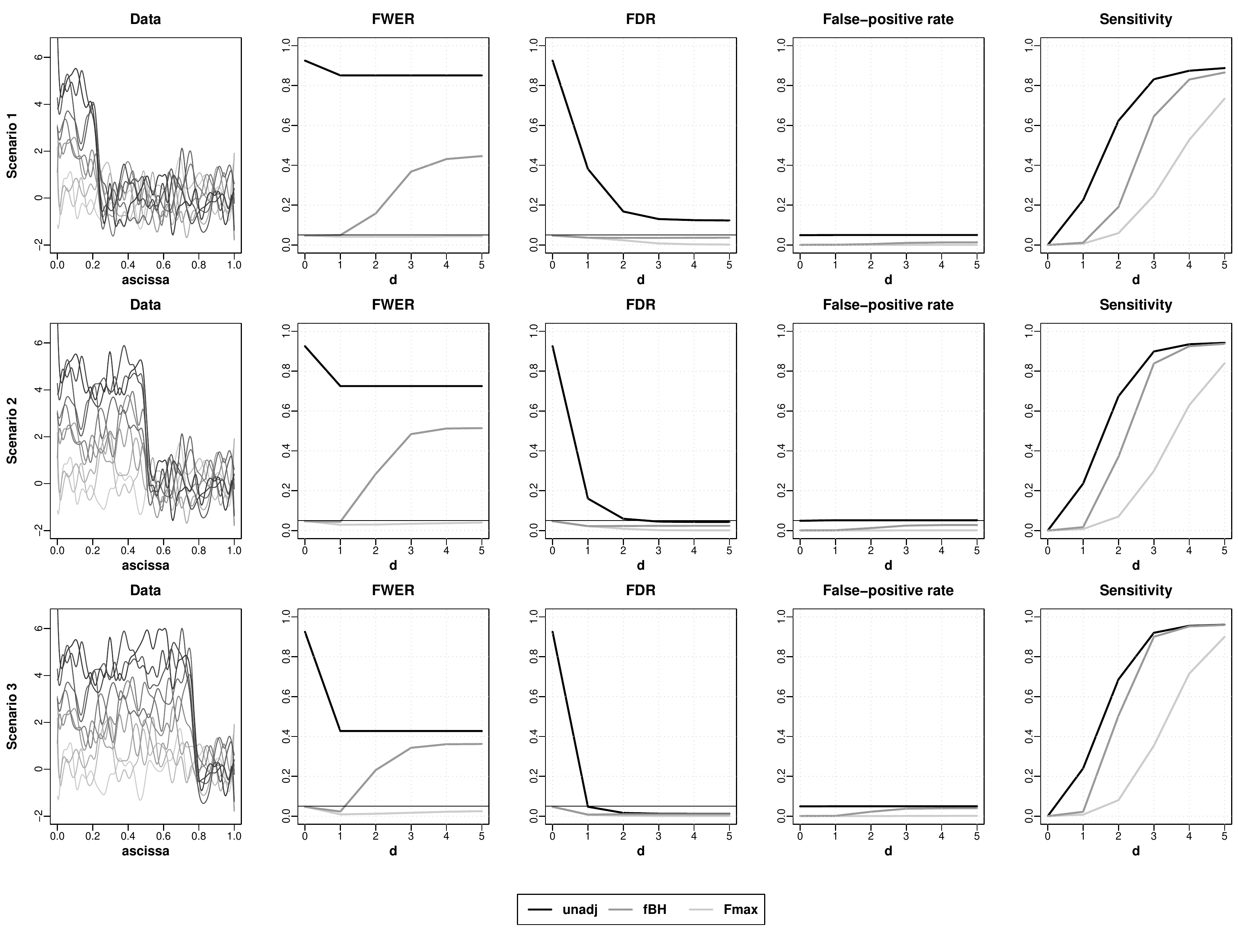}
    \caption{Simulation study: Familywise error rate, false discovery rate, false positive rate and sensitivity analysis of three methods for varying values of $d$. Scenario 1 corresponds to $h=10$, Scenario 2 to $h=20$, and Scenario 3 to $h=30$. 'Data' shown in the left panels are examples of simulated data for $d=1$.}
    \label{fig:simulation_1D}
\end{figure}

\subsection{2D-simulation} \label{afsnit-simulation}
\paragraph{Description of simulation}

The base signal $\theta$ for this simulation  consisted of 9 conical spikes with height $h = 1$ and diameter $d = 0.2$ arranged on the grid $\{0.25, 0.50, 0.75\}^2$ with the unit square $\mathbf{T} = (0,1)^2$ as doma
in. Five spikes had positives values and the remaining four spikes had negative values. 
On top we added an error signal generated as a smooth Matérn field $x_i$ with varying scale parameter.  The observed signal is $y_i = \theta + x_i$ for $i = 1, \dots, N$. See Figure \ref{sim-fields} for an illustration.

\begin{figure}
	\centering 
	\includegraphics[width=\textwidth]{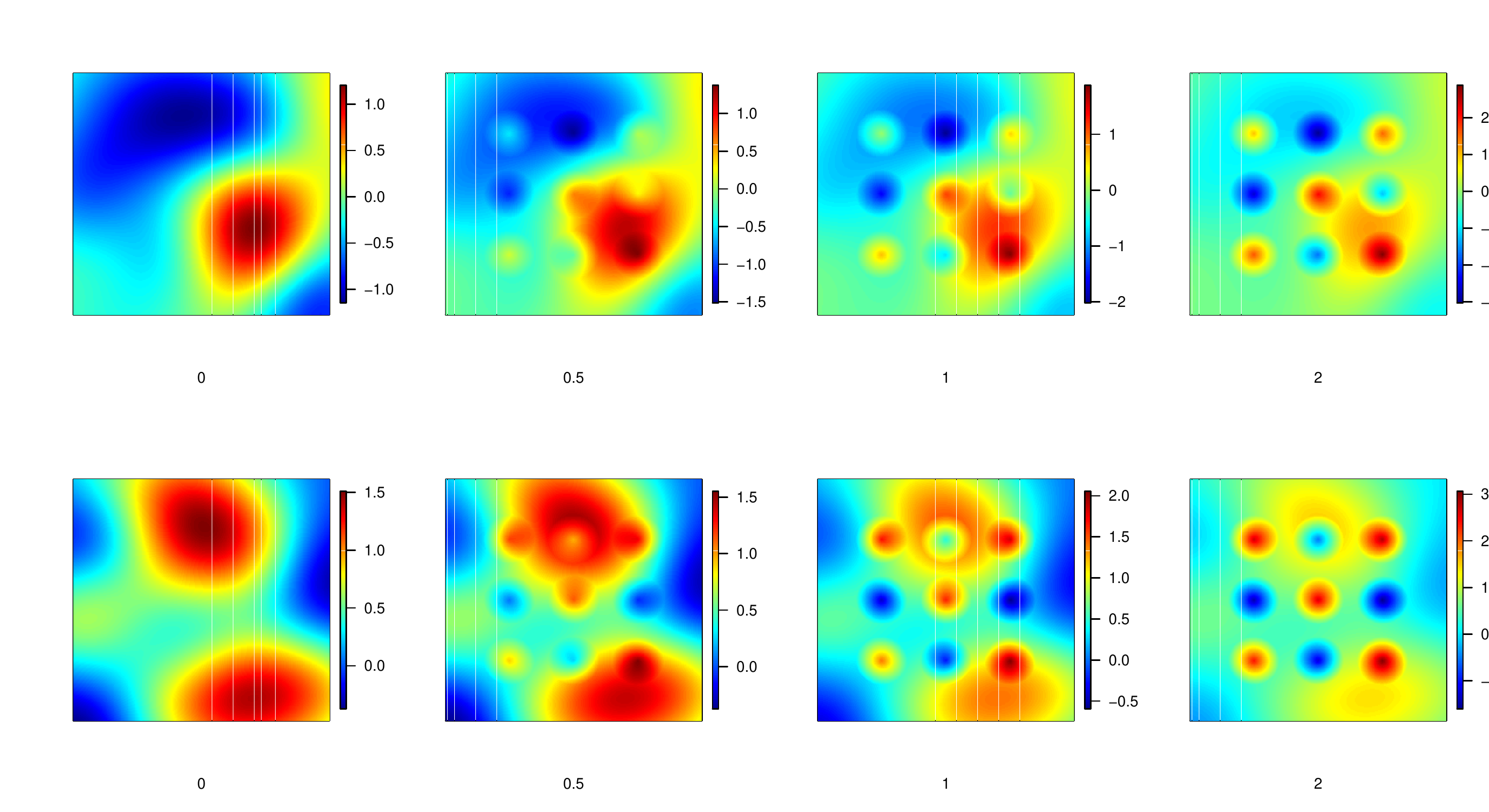}
	\caption{Two simulated fields with increasing signal strength} \label{sim-fields}
\end{figure}

The simulation was inspired by \cite{cheng2017}, who also studies FDR in a setting of random fields, albeit with a quite different scope. 
We tested against the pointwise null hypothesis $H_0(t): \theta(t) = 0$ using a pointwise, one-sided t-test over a fine lattice.

The Benjamini-Hochberg procedure was applied, and false positive rate (FPR), false discovery rate and sensitivity values were evaluated, along with false discovery rate for the unadjusted p-values for comparison. We simulated 2500 samples of the error process at grid $255 \times 255$. 
We assessed performance of the test in five different setups with varying strengths of the base signal and varying numbers of samples per t-test.
Samples of the error process were recycled and used for all experimental setups, this is also the reason for the varying number of replications in table \ref{tabel-fem-pct}.

We remark that the observed data will show a "bend" at the edges of the spikes, making these easily detectable by other methods. However, the use of pointwise tests will ignore such features.

\paragraph{Simulation results}
We tested at various thresholds, $\alpha \in \{0.001, 0.01, 0.02, 0.03, 0.04, 0.05, 0.10\}$ for five different experimental setups described in Table \ref{tabel-fem-pct}; results are shown in Figure \ref{sim-results-gg}. 
The sensitivity values and false postive rate (FPR) are defined as proportions of rejected/accepted hypotheses to the total number of false and correct hypotheses, respectively, after p-value adjustment. 
$\mu(U) = 71.7\%$ of the signal was zero, thus we would expect FDR to be controlled by $0.717 \alpha$; for the 5\% threshold of table 1, this is roughly $0.036$. 

\begin{table} \centering
	\begin{tabular}{ccc|Lcccc}
			$\begin{array}{c} \text{Setup} \\ \text{no.}	\end{array}$ & $\begin{array}{c} \text{Signal} \\ \text{size}	\end{array}$ & $\begin{array}{c} \text{\# of samples} \\ \text{per test}	\end{array}$ & Re\-pli\-ca\-tions & Sensi\-tivity & FPR & FDR & $\begin{array}{c} \text{FDR,} \\ \text{unadjusted}	\end{array}$ \\ \hline 
			1 & 2.0 & 20 & 125 & 0.498 & 0.00785 & 0.0252 & 0.113\\
			2 & 2.0 & 10 & 250 & 0.226 & 0.00553 & 0.0257 & 0.135\\
			3 & 2.0 & 40 & 62 & 0.654 & 0.00810 & 0.0242 & 0.117\\
			4 & 1.0 & 20 & 125 & 0.116 & 0.00467 & 0.0249 & 0.155\\
			5 & 0.5 & 20 & 125 & 0.0065 & 0.00295 & 0.0258 & 0.235
	\end{tabular} \caption{Results from simulation at 5\% threshold with 255x255 sample points. The varying number of replications is due to the upper limit of 2500 signals in total}  
	\label{tabel-fem-pct}
\end{table}
\begin{figure} 
	\centering
	\includegraphics[width=\textwidth]{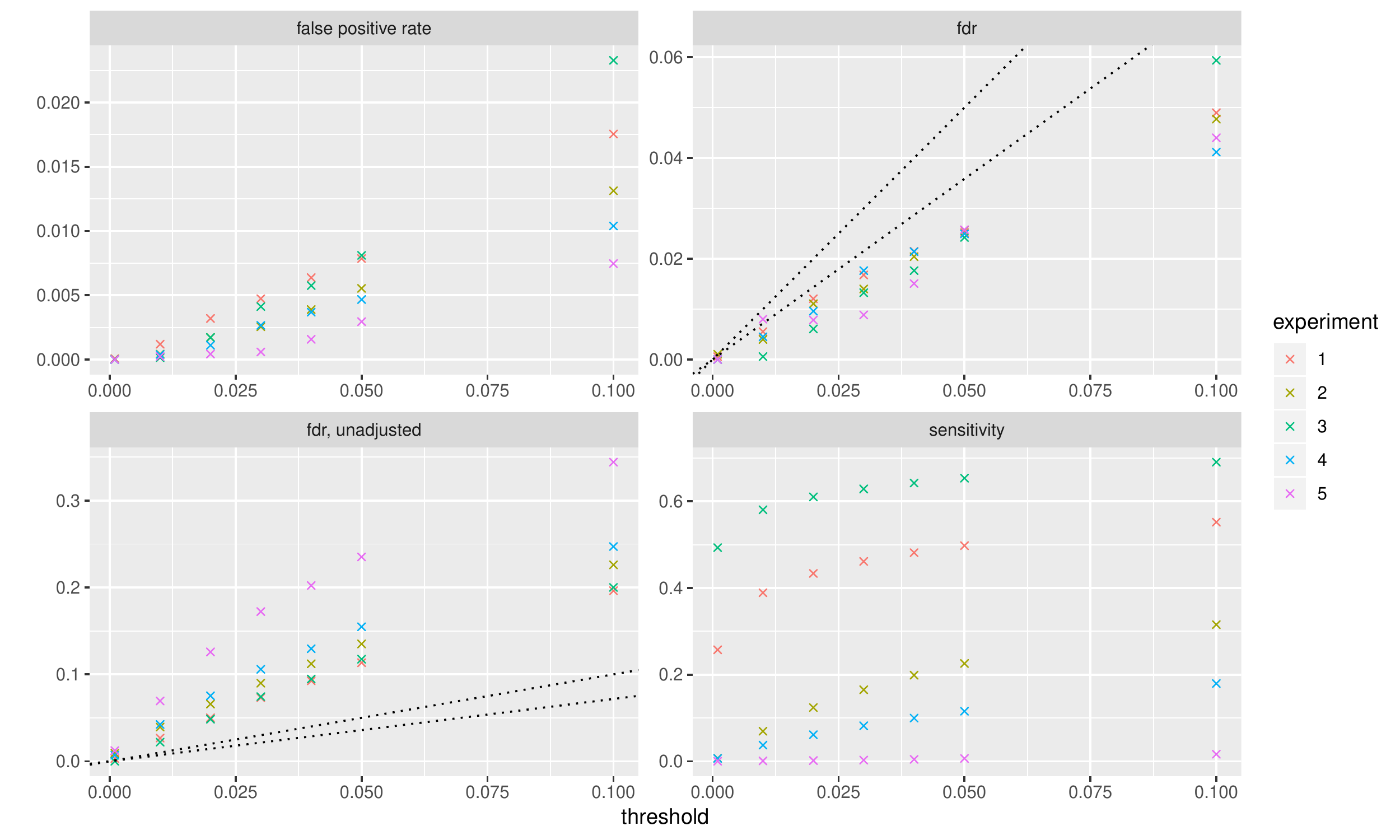}
	\caption{Average sensitivity, false positive rate (FPR) and FDR for different levels of $\alpha$ in five experimental setups as described in Table \ref{tabel-fem-pct}, along with FDR for the unadjusted $p$-values. The dotted lines have slopes $1$ and $0.717$, respectively.} \label{sim-results-gg}
\end{figure}

Unsurprisingly, sensitivity (power level), FPR and FDR increase with $\alpha$. Experiment 5 has almost no power, even at $\alpha = 0.10$, but still a comparatively large FDR, indicating that the Benjamini-Hochberg procedure is not too conservative in this setup. The false discovery rate is remarkably stable across experimental setups, and shows a linear tendency. As expected, FDR is well below $\alpha$ in all instances, and also below $\mu(U) \alpha$. In comparison, the FDR for the unadjusted p-values exceeds $\alpha$ in all setups except one, and varies considerably with the experimental setting.

\section{Application: Analysis of Climate Data} \label{klimadata-afsnit}

Climate change is a huge issue, both politically and scientifically. 
The main issue are increasing temperatures with many adverse effects on weather and climate. 
Knowing that temperature has increased significantly on a global scale, we wanted to test where on Earth temperature has increased.

\subsection{Data and model}

Data consists of yearly averages of temperatures, starting in 1983 and ending in 2007, for each $1^\circ \times 1^\circ$ tile on Earth, using standard latitudes and longitudes. Temperatures are satellite measurements collected by NASA. These data were obtained from the NASA Langley Research Center Atmospheric Science Data Center Surface meteorological and Solar Energy (SSE) web portal supported by the NASA LaRC POWER Project. \footnote{\url{http://eosweb.larc.nasa.gov}}

One crucial feature is that data was more densely sampled closer to the poles than close to Equator. Naturally, data also exhibits behaviour depending on local geography. Rather than viewing data as truly areal, we considered data to sit on the midpoints of the tiles, e.g. $(73.5^\circ S, 24.5^\circ W)$ corresponds to the tile $(74^\circ S, 73 ^\circ S) \times  (24^\circ W, 25 ^\circ W)$.

We applied the linear regression model $y_{st} = a_s + b_s \text{year}_t + \epsilon_{st}, s \in S^2, t \in \{1983, \dots, 2007 \}$, testing for \emph{positive trend}, i.e. $H^0_s: b_s = 0$ with alternative hypothesis $H^A_s: b_s > 0$.
Figure \ref{fig-temp} displays the temperature changes,  values of the pointwise t-statistics and corresponding (unadjusted) p-values. Observe how much the test statistic varies across the globe, and how differences between land and ocean are more visible in the lower the plot compared to the upper plot.
\begin{figure}[!htbp]
	\centering
	\includegraphics[width=0.7\textwidth]{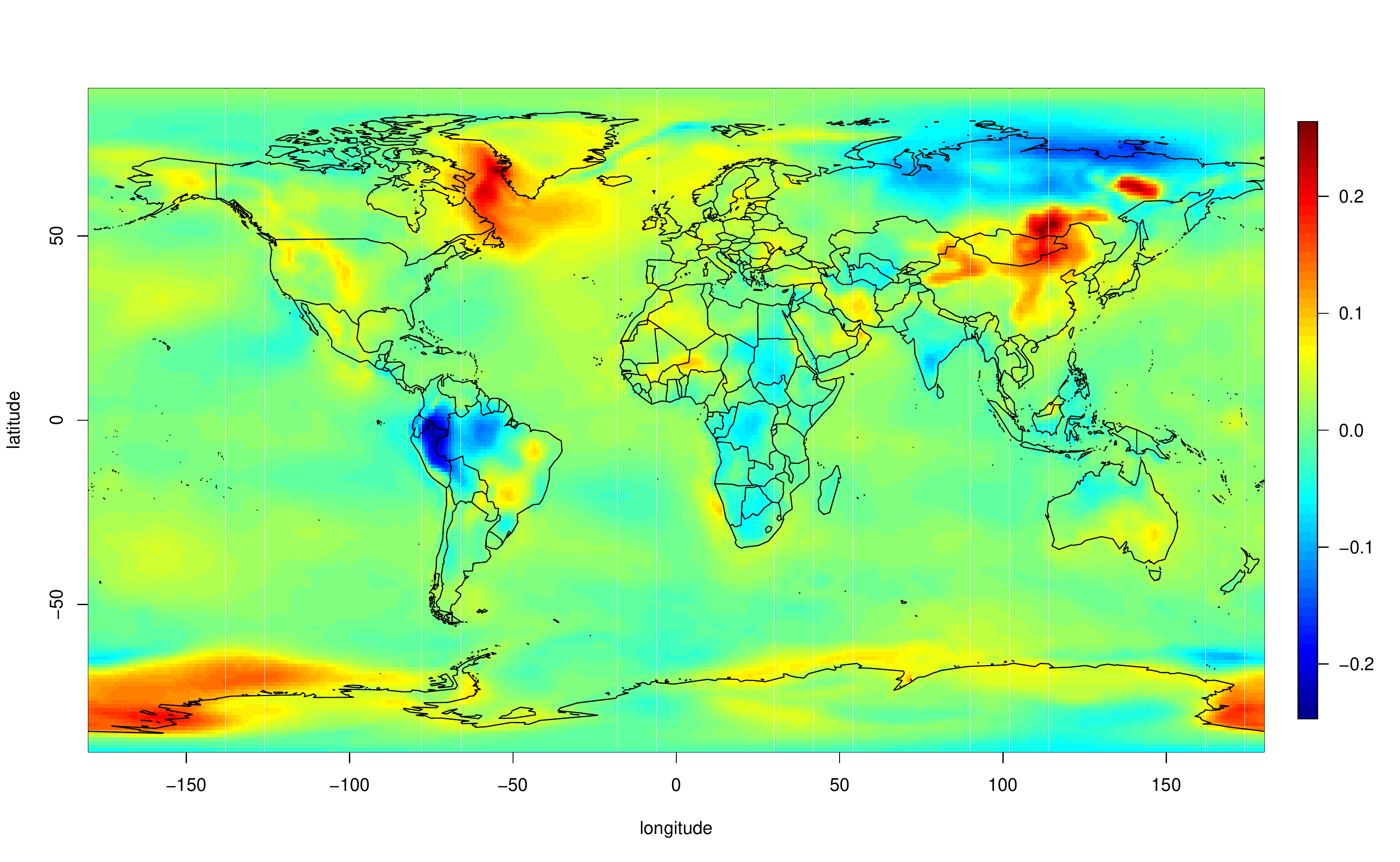}
	\includegraphics[width=0.7\textwidth]{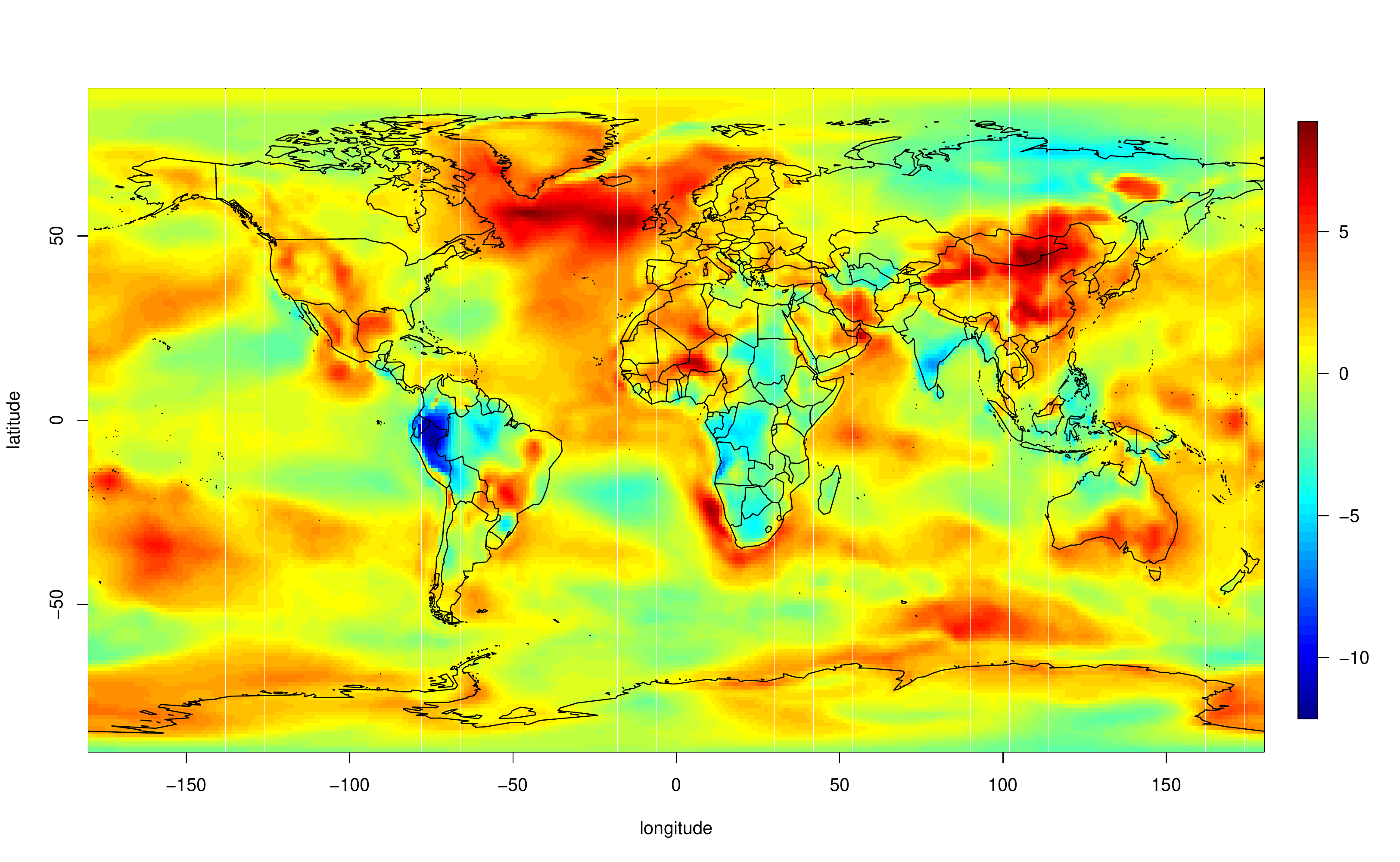}
	\includegraphics[width=0.7\textwidth]{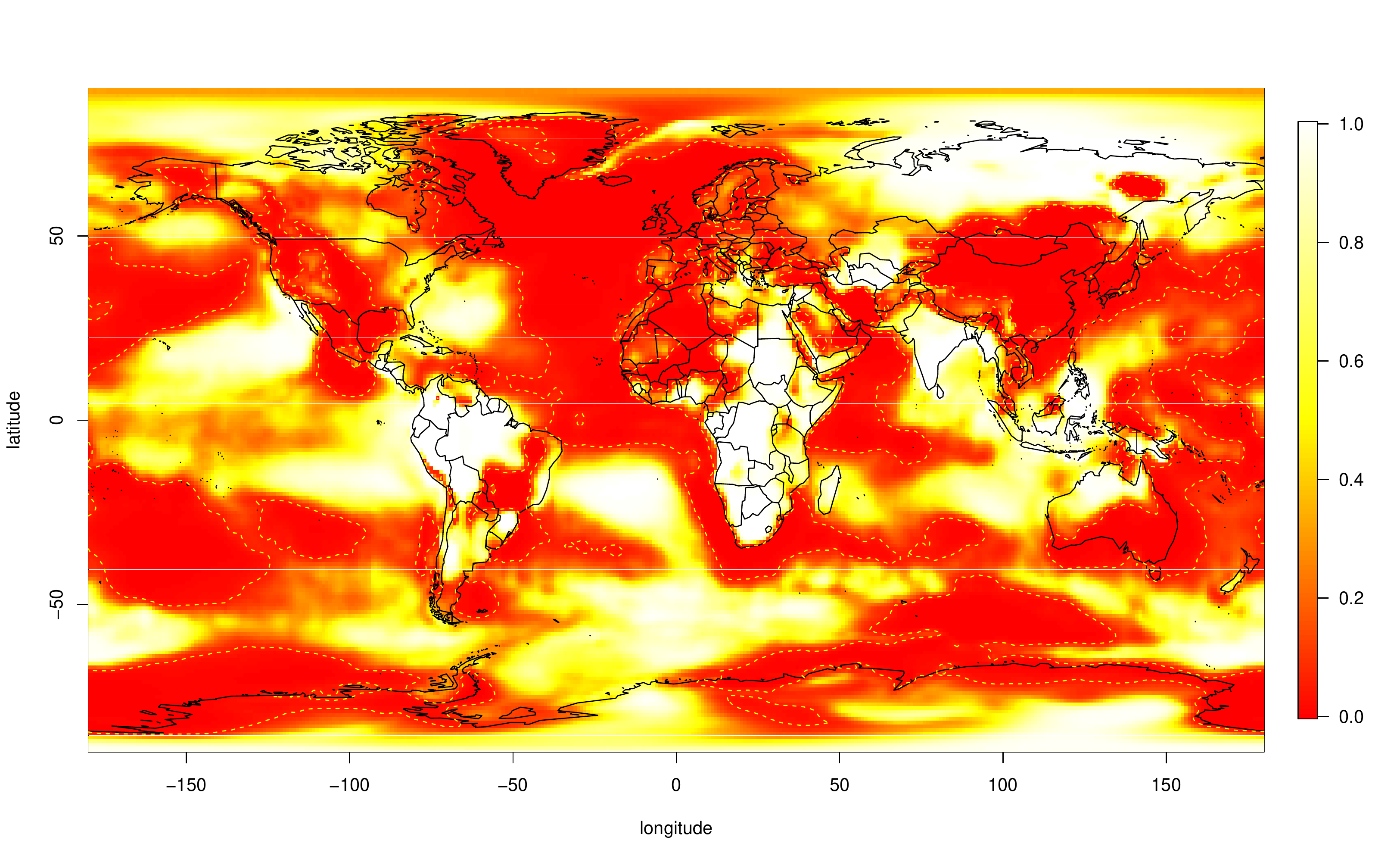}
	\caption{Above: Average yearly temperature change in centigrades, 1983-2007. Middle: Pointwise values of test statistic ({t-distribution}, 24 degress of freedom). Below: Unadjusted p-values with marking of the 5\% threshold. \label{fig-temp}
	}
\end{figure}

To perform the BH procedure, we mapped the sphere into $ \mathbf{T} = (-\pi, \pi) \times (-\pi/2, \pi/2)$ by (scaled) polar coordinates, ie. longitude and latitude. This mapping gives rise to a measure $\nu = f \cdot \mu$ on $\mathbf{T}$ where $f$ is proportional to $\cos($latitude) cf. Remark \ref{mangfoldighedsremark}. This measure gives uniform weights to all points on Earth, assuming Earth to be a perfect sphere.
One-sided t-tests were used for obtaining unadjusted p-values. We used the same grid as the observations for approximating the BH procedure, in total $180\times360=64800$ grid points.

\subsection{Results}
\begin{figure}[!htbp]
	\centering
	\includegraphics[width=0.7\textwidth]{graphs/unadjusted_pvals.pdf}
	\includegraphics[width=0.7\textwidth]{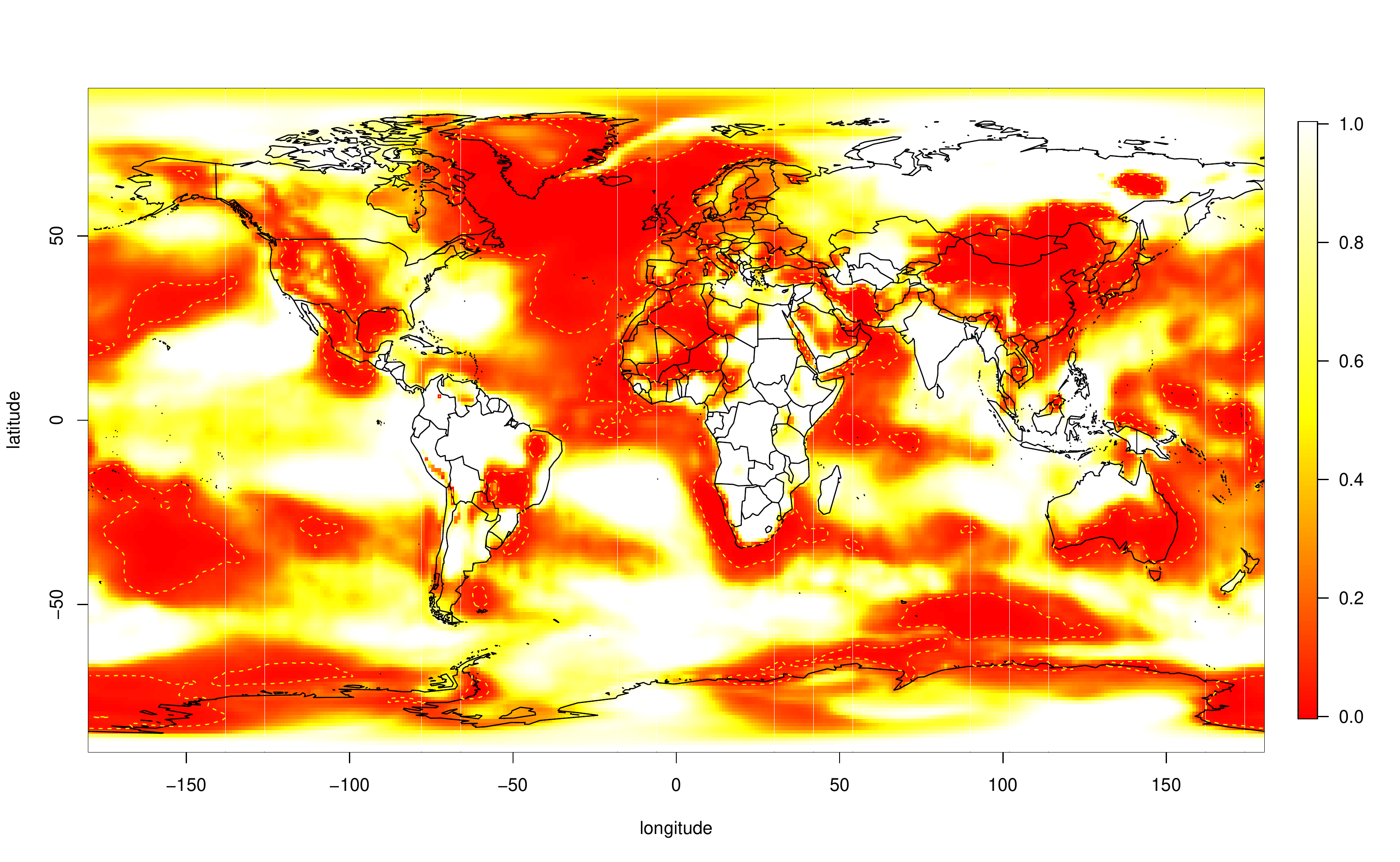}
	\caption{Upper plot: Unadjusted p-values. 
	Lower plot: FDR-adjusted p-values. Dashed lines indicate 5\% significance levels.} \label{fig-iwt}
\end{figure}
The coverage areas at various significance levels are provided in Table \ref{tabel-sig}. 

Although more conservative by construction than unadjusted p-values, the fBH-adjusted p-values still retained large significant areas,  indicating that the temperature increase observed in these areas is very unlikely to be a coincidence in but a small fraction of these areas.
If we take look at the map, the North Atlantic Ocean and northern China stands out; it is evident that these regions have experienced  temperatures far above the normal in the latest years with the adverse weather effects this may cause.

\begin{table} {\centering
		\begin{tabular}{c|ccc} 
			Significance level & Unadjusted p value& fBH-adjusted p value\\ \hline 
			0.10  & 0.410 & 0.229 \\
			0.05  & 0.324 & 0.150 \\
			0.01 & 0.178  & 0.062 \\
			0.001 & 0.074 & 0.023
		\end{tabular} \caption{Areas of significance of the correction methods at various significance levels as percentage of Earth total} \label{tabel-sig} }
\end{table}

\section{Discussion} \label{diskussions-afsnit}

In this paper we  defined the false discovery rate for functional data defined on generic subdomains of $\R^k$, and by remark \ref{mangfoldighedsremark}, this is easily extended to non-euclidean domains such as the sphere used in the data application. Furthermore, we devised a correction method for controlling FDR, the \emph{functional BH procedure}, along with an adjusted p-value function. We showed a finite dimensional approximation of the fBH procedure, and proved that it converges to its theoretical infinite-dimensional definition.
The fBH procedure was successfully applied in two simulations and on a data set on climate change. 

The two simulation studies allowed us to study the fBH procedure in a setting where the distributions and true null hypotheses were known. The false discovery rate was controlled by FBH procedure in all instances, unlike the unadjusted p-value functions. The signal-to-noise ratio proved to be important: the sensitivity of the fBH method increased with signal strength, but the false discovery rate was remarkably constant.

By means of the climate change analysis, we demonstrated the applicability of the method and also gained insight into which regions of Earth that have seen temperature increases due to global warming in a recent time span. More advanced models and tests may further increase our understanding of local temperature changes in connection to warming, but we leave this as future work.

The fBH procedure is easily applicable in practice thanks to its finite dimensional approximation. However, it should be noted that as generally is the case in functional data analysis, the performances of the procedure depends on how well the discrete grid approximates the functional data on one hand, and how well the functions are smoothed from data on the other hand.

We would like to stress the minimal assumptions required of the fBH approach: the dependence structure of functional data such as the Earth climate data can be complex and difficult to model. However, our approach  does not require specific modelling of the covariance structure of the data; we merely require a certain degree of positive association among $p$-values. 

Functional FDR and the fBH procedure suffer from some of the same issues as their corresponding discrete versions, see. e.g. \cite{storey2003pfdr}. As noted in Section \ref{lit-review}, many authors have proposed methods or quantities to deal with multiple testing. Given the success of formulating FDR and the BH procedure in a functional framework, it is likely that some of these other methods or quantities can be expanded to the functional case as well.

Due to its simple applicability, general setting and easy understanding, we expect fFDR to have a great potential as a tool for local inference in functional data analysis. The functional BH procedure require only little computational power, and should at most be a minor issue in applications.



\clearpage
 \appendix
{ \noindent \huge \textbf{Appendix}}

\section{Proofs}

\textbf{Proof of Proposition \ref{hoved-thm} and Proposition \ref{general-thm}}

\subsection{Assumptions}
We begin by repeating the assumptions of Proposition \ref{hoved-thm} and Proposition \ref{general-thm}.

\begin{defi*}[PDRS]
Let '$\leq$' be the partial/usual ordering on $\R^l$.
	An \emph{increasing set} $D \subseteq \R^l$ is a set satisfying $x \in D \wedge y \geq x \Rightarrow y \in D$.
	
	A random variable $\mathbf{X}$ on $\R^l$ is said to be \emph{PDRS on  $I_0$},  where $I_0$ is a subset of $\{1, \dots, l\}$, if it for any increasing set $D$ and $i \in I_0$ holds that 
\begin{equation*} 
	x \leq y \Rightarrow P(\mathbf{X} \in D | X_i = x) \leq P(\mathbf{X} \in D | X_i = y)
\end{equation*}

Let  $\mathbf{Z}$ be an  infinite-dimensional random variable, where instances of $\mathbf{Z}$ are functions $T \pil \R$. We say that $\mathbf{Z}$ is PDRS on $U \subseteq T$ if all finite-dimensional distributions of $\mathbf{Z}$ are PDRS. That is, for all finite subsets $I = \{i_1, \dots , i_l \} \subseteq T$, it holds that
 $Z(i_1), \dots, Z(i_l)$ is PDRS on $I \cap U$.
\end{defi*}
Let $\{S_k\}_{k=1}^\infty, S_1 \subset S_2 \subset \dots$ be a dense, uniform grid in $\bt$ in sense that $S_k$ uniformly approximates all level sets of $p$ and $p|_U$ with probability one.

For Theorem \ref{hoved-thm} that amounts to,
\begin{equation}
P \left[ \lim_{k \pil \infty} \sup_r \frac{\#( S_k \cap \{s: p(s) \leq r \})} {\# S_k} -  \mu\{s: p(s) \leq r \} \pil 0 \right] = 1 \label{th12-approx1}
\end{equation} and
\begin{equation}
P \left[ \lim_{k \pil \infty} \sup_r \frac{\#( S_k \cap \{s: p(s) \leq r \} \cap U)} {\# S_k} -  \mu(\{s: p(s) \leq r \} \cap U) \pil 0 \right] = 1 \label{th12-approx2}
\end{equation}
whereas for  Theorem \ref{general-thm}, we need the density function $f$:
\begin{equation*}
P \left[ \lim_{k \pil \infty} \sup_r \frac{\sum_{i \in S_k \cap \{s: p(s) \leq r \}} f(i) } {\# S_k} -  \int_{\{s: p(s) \leq t \}} f(x) \de x  \pil 0 \right] = 1
\end{equation*} and
\begin{equation*}
P \left[ \lim_{k \pil \infty} \sup_r \frac{\sum_{i \in S_k \cap \{s: p(s) \leq r \} \cap U} f(i) } {\# S_k} -  \int_{\{s: p(s) \leq t \} \cap U} f(x) \de x  \pil 0 \right] = 1
\end{equation*}

Furthermore, assume that $p$ is PDRS wrt. the set of true null hypotheses with probability one, and that the assumptions about $p$-value function below hold true with probability one:
\begin{enumerate}
	\item[(a1)] All level sets of $p$ have zero measure, 
	\[ \mu\{s: p(s) = t \} = 0 \quad \forall t \in \bt
	\]
	\item[(a2)] $\alpha^* \in (0,\alpha] \Rightarrow$: for any open neighbourhood $O$ around $\alpha^*$ there exists $s_1, s_2 \in O$ s.t. $ a(s_1) > \alpha^{-1} s_1, a(s_2) < \alpha^{-1} s_2$, where $a$ is the cumulated p-value function (Definition \ref{bh-def-unweighted}).
	\item[(a3)]  $[\alpha^* = 0] \Rightarrow \min p(t) > 0$. 
\end{enumerate}

\subsection{Proof details}
For the ease of presentation, we will only consider Theorem \ref{hoved-thm} and furthermore assume that $\mu(\bt) = 1$; the latter can be done without loss of generalisation. The proof of Proposition \ref{general-thm} is analogous but notionally tedious, as the counts are replaced by sums and the measures by integrals. 

Let $a_k$ be the cumulated p-value function for the $k$'th iteration of the BH procedure: 
\begin{equation*}
a_k(t) := N_k \# \{ s \in S_k : p(s) \leq t\}
\end{equation*}
and define the k'th step false discovery proportion $Q_k$  by applying the (usual) BH procedure at level $q$ to $p$ evaluated in $S_k$:
\begin{equation*}
    Q_k = \frac{\#\{t \in S_k : p(t) \leq b_k\} \cap U}{\#\{t \in S_k : p(t) \leq b_k\}}, \quad 
    b_k = \arg\max_r \frac{\# \{s \in S_k : p(s) \leq r \}}{\#S_k} \geq \alpha^{-1} r 
\end{equation*}
or equivalently $b_k = \arg\max_t a_k(t) \geq \alpha^{-1} r$.

\begin{lemma} \label{lemma-0v2}
	$a_k$ converges to $a$ uniformly as $k \pil \infty$. 
\end{lemma}
Proof: Follows from assumption \eqref{th12-approx1} and definitions of $a_k$ and $a$.

\begin{lemma} \label{lemma-et} $b_k$ converges to $\alpha^*$ as $k \pil \infty$
\end{lemma}

\begin{proof}
	By Lemma \ref{lemma-0v2}, $a_k$ converges uniformly to $a$. There are two cases: $\alpha^* = 0$  and $\alpha^* \in (0, \alpha]$. 
	
	\paragraph{Case 1, $\alpha^* = 0$} 
	
	Let $O$ be any open neighbourhood around zero. $O^{C}$ (where the complement is wrt. [0,1]) is a closed set that satisfies $a(t) < \alpha^{-1}t $. By continuity of $a$, there exists an $\epsilon > 0$ s.t. $a(t) < \alpha^{-1}t - \epsilon$ for all $t \in O^C$. As $a_k$ converges uniformly to $a$, eventually for large enough $k$, $a_k(t) < \alpha^{-1}t$ for $t \in T \backslash O$, and thus $b_k \in O$ eventually. This was true for any $O$, and we conclude that $b_k \pil 0$. 
	
	\paragraph{Case 2, $\alpha^* \in (0, \alpha]$}
	
	By assumption, for any open neighbourhood $O \ni \alpha^*$, there exist $s_1, s_2 \in O$ s.t. $a(s_1) > \alpha^{-1}s_1$, $a(s_2) < \alpha^{-1}s_2$. 
	
	For $t > \alpha^*, t \notin O$, we have that $\alpha^{-1} t - a(t) > \epsilon $ for some $\epsilon > 0$ by continuity of $a$. Hence by uniform convergence, it must hold that for $k$ sufficiently large we have $a_k(t) < \alpha^{-1}t$ for $t > \alpha^*, t \notin O$. This was true for any $O$, and we conclude $\limsup b_k \leq \alpha^*$. 
	
	Conversely, we can show that $\liminf b_k \geq \alpha^*$, and thus $\lim b_k = \alpha^*$.

\end{proof}

Define $A_k = \#\{t \in S_k : p(t) \leq b_k \}$, and define $Q_k$ as the false discovery proportion for the $k$'th iteration:
\begin{equation*}
Q_k := \frac{\# (A_k \cap U)}{\# A_k}1_{A_k \neq \emptyset}
\end{equation*}

\paragraph{Rejection areas} Now we intend to prove 	that $H_{t,k}$ converges eventually.  Note that $p(t)$ is independent of $k$, and that $H_{t,k} = (t \in S_k) \cap (p(t) \leq b_k)$, i.e. the event that the BH threshold at step $k$ is larger than $p(t)$. 

\begin{prop} \label{htk-konv}
	For all $t$ that satisfies $p(t) \neq \alpha^*$, $H_{t,k}$ converges eventually. 
\end{prop} 

\begin{proof}
	First note that if $t \notin S_k$ for all $k$, then $H_{t,k}$ is trivially zero for all $k$. 
	So assume $t \in S_{k_0}$ for some $k_0$. As $k \pil \infty$, $b_k \pil \alpha^*$, and by assumption $p(t) \neq \alpha^*$. Eventually, as $k \pil \infty$, $p(t)$ is either strictly larger or strictly smaller than $b_k$, proving the result.  
\end{proof}

\paragraph{Convergence of $Q_k$}
Finally we need to show that $Q_k \pil Q$. We show this by proving convergence of the nominator and denominator, and arguing that $Q = 0$ implies that $Q_k$ is 0 eventually.

Define $H^0 = \{t | p(t) > \alpha^*\}$, ie. the acceptance region, and $H^1 = T \backslash H^0$, the rejection region. Note that $\mu(H^1) = a( \alpha^*) = \alpha^{-1} \alpha^* $.
%
Also note that $H^1 = V \cup S = \{t: p(t) \leq \alpha^*\}$ and $H^1 \cap U = V$. 


\begin{prop} \label{prop-qk}
	$ N_k \#  A_k   \pil \mu(H^1)$ and 
	$N_k \# (A_k \cap U) \pil \mu(H^1 \cap U)$.
\end{prop}
\begin{proof}

For $k$, define $J_k = \{t : p(t) \leq b_k \}$. Note that $A_k = J_k \cap S_k$.
%
Observe that by the assumption about uniform convergence on levels sets (Eq. \eqref{th12-approx1}):
\[
N_k \# (J_k \cap S_k) - \mu(J_k) \pil 0 \for k \pil \infty.
\]

Next observe that due to (1) continuity of $a$, (2) $b_k \pil \alpha^*$  and (3) the fact that we are considering sets on the form $\{t: p(t) \leq x \}$, we are able to conclude that
\begin{equation*}
\mu(J_k \triangle H^1)  \pil 0 \for k \pil \infty.
\end{equation*}
and we conclude $N_k \#  A_k \pil \mu(H^1)$.

For the second part, observe that from Eq. \eqref{th12-approx2}) it follows that
\[
N_k \# (U \cap J_k \cap S_k) - \mu(U \cap J_k) \pil 0 \for k \pil \infty.
\]
We just argued that $\mu(J_k \triangle H^1) \pil 0$. It remains true when "conditioning"\ on a measurable set, in this case $U$: 
\begin{equation*}
\mu((J_k \cap U) \triangle (H^1 \cap U))  \pil 0 \for k \pil \infty.
\end{equation*}
and we conclude $N_k \#  (A_k \cap U) \pil \mu(H^1)$.
\end{proof}

For $\alpha^* = 0$ we have the following stronger result:
\begin{lemma}\label{a0-lemma}
	If $\alpha^* = 0$, then $\# A_k = 0$ eventually. 
\end{lemma}
From this lemma it follows that $ N_k \#  A_k = 0$ (and thus $Q_k$ as well) eventually.
\begin{proof}
	Since $\alpha^* = 0$, $a(t) < \alpha^{-1}t$ for all $t> 0$. By assumption, $\min p(t) > 0$, and thus $a_k(s) = 0$ for $s < \min p(t)$ and all $k$. 
	
	By continuity of $a$, it follows that there exists $\epsilon > 0$ s.t. $\alpha^{-1} t - a(t) > \epsilon$ on the interval $[\min p(t) , 1]$, and by uniform convergence of $a_k$ we get that for large enough $k$, $a_k(t) < \alpha^{-1} t$ for all $t \geq \min p(t)$. 
	
	Combining with 	$a_k(t) = 0$ for $t < \min p(t)$, we get that eventually $a_k(t) < \alpha^{-1}t$ for every $t > 0$ and thus $b_k = 0$. From this (remember $\min p(t)>0$) we conclude that all hypotheses are rejected eventually, ie.  $\# A_k = 0$ for $k$ sufficiently large. 
\end{proof}

\begin{thm}  \label{qk-konv-1}
	$Q_k$ converges to $Q$ almost surely, and $\limsup_{k \pil \infty} \E[Q_k] \leq \alpha \mu(U)$.
\end{thm}
\begin{proof}
	By Lemma \ref{a0-lemma}, $Q_k$ converges to $Q$ when  $\alpha^* = 0$, and by Proposition \ref{prop-qk}  $Q_k$ converges to $Q$ when $\alpha^* > 0$ since $\mu(H^1) = \alpha^*/\alpha > 0$.
	
	Applying Benjamini and Yekuteli's original proposition, Theorem \ref{pdrs-thm}, (now we use the PDRS assumption), we have $\E[Q_k] \leq \alpha N_k \#(S_k \cap U)$ for all $k$. 
	By setting $r = 1$ it follows from \eqref{th12-approx2} that $\lim_{k \pil \infty} N_k \#(S_k \cap U) = \mu(U)$, and hence  $\limsup_{k \pil \infty} \E[Q_k] \leq \alpha \mu(U)$.
	
\end{proof}



\begin{thebibliography}{}
	\ifx \url   \undefined \def \url#1{#1}   \fi
	
	\bibitem{abramowiczmox}
	\textsc{Abramowicz, K.}, \textsc{H{\"a}ger, C.~K.}, \textsc{Pini, A.},
	\textsc{Schelin, L.}, \textsc{Sj{\"o}stedt~de Luna, S.}, \textsc{and}
	\textsc{Vantini, S.} (2018).
	\newblock Nonparametric inference for functional-on-scalar linear models
	applied to knee kinematic hop data after injury of the anterior cruciate
	ligament.
	\newblock \emph{Scandinavian Journal of Statistics\/}.
	
	\bibitem{BH1995}
	\textsc{Benjamini, Y.} \textsc{and} \textsc{Hochberg, Y.} (1995).
	\newblock Controlling the false discovery rate: a practical and powerful
	approach to multiple testing.
	\newblock \emph{Journal of the royal statistical society. Series B
		(Methodological)\/}, 289--300.
	
	\bibitem{benjamini1997}
	\textsc{Benjamini, Y.} \textsc{and} \textsc{Hochberg, Y.} (1997).
	\newblock Multiple hypotheses testing with weights.
	\newblock \emph{Scandinavian Journal of Statistics\/}~\textbf{24},~3, 407--418.
	
	\bibitem{BY2001}
	\textsc{Benjamini, Y.} \textsc{and} \textsc{Yekutieli, D.} (2001).
	\newblock The control of the false discovery rate in multiple testing under
	dependency.
	\newblock \emph{Annals of statistics\/}, 1165--1188.
	
	\bibitem{berlinetAgnan}
	\textsc{Berlinet, A.} \textsc{and} \textsc{Thomas-Agnan, C.} (2011).
	\newblock \emph{Reproducing kernel Hilbert spaces in probability and
		statistics}.
	\newblock Springer Science \& Business Media.
	
	\bibitem{cheng2017}
	\textsc{Cheng, D.}, \textsc{Schwartzman, A.}, \textsc{and} \textsc{others}.
	(2017).
	\newblock Multiple testing of local maxima for detection of peaks in random
	fields.
	\newblock \emph{The Annals of Statistics\/}~\textbf{45},~2, 529--556.
	
	\bibitem{efron2001etal}
	\textsc{Efron, B.}, \textsc{Tibshirani, R.}, \textsc{Storey, J.~D.},
	\textsc{and} \textsc{Tusher, V.} (2001).
	\newblock Empirical bayes analysis of a microarray experiment.
	\newblock \emph{Journal of the American statistical
		association\/}~\textbf{96},~456, 1151--1160.
	
	\bibitem{freedman1983nonstochastic}
	\textsc{Freedman, D.} \textsc{and} \textsc{Lane, D.} (1983).
	\newblock A nonstochastic interpretation of reported significance levels.
	\newblock \emph{J. Bus. Econ. Stat.\/}~\textbf{1},~4, 292--298.
	
	\bibitem{heesen2015}
	\textsc{Heesen, P.}, \textsc{Janssen, A.}, \textsc{and} \textsc{others}.
	(2015).
	\newblock Inequalities for the false discovery rate (fdr) under dependence.
	\newblock \emph{Electronic Journal of Statistics\/}~\textbf{9},~1, 679--716.
	
	\bibitem{holmes1996}
	\textsc{Holmes, A.~P.}, \textsc{Blair, R.}, \textsc{Watson, J.}, \textsc{and}
	\textsc{Ford, I.} (1996).
	\newblock Nonparametric analysis of statistic images from functional mapping
	experiments.
	\newblock \emph{Journal of Cerebral Blood Flow \& Metabolism\/}~\textbf{16},~1,
	7--22.
	
	\bibitem{horvathK2012}
	\textsc{Horv{\'a}th, L.} \textsc{and} \textsc{Kokoszka, P.} (2012).
	\newblock \emph{Inference for functional data with applications}.
	Vol.~\textbf{200}.
	\newblock Springer Science \& Business Media.
	
	\bibitem{pv2017interval}
	\textsc{Pini, A.} \textsc{and} \textsc{Vantini, S.} (2017).
	\newblock Interval-wise testing for functional data.
	\newblock \emph{Journal of Nonparametric Statistics\/}~\textbf{29},~2,
	407--424.
	
	\bibitem{Ramsay}
	\textsc{Ramsay, J.~O.} \textsc{and} \textsc{Silverman, B.~W.} (2005).
	\newblock \emph{Functional Data Analysis}, Second ed.
	\newblock Springer.
	
	\bibitem{schwartzman2011}
	\textsc{Schwartzman, A.}, \textsc{Gavrilov, Y.}, \textsc{and} \textsc{Adler,
		R.~J.} (2011).
	\newblock Multiple testing of local maxima for detection of peaks in 1d.
	\newblock \emph{Annals of statistics\/}~\textbf{39},~6, 3290.
	
	\bibitem{storey2003pfdr}
	\textsc{Storey, J.~D.} \textsc{and} \textsc{others}. (2003).
	\newblock The positive false discovery rate: a bayesian interpretation and the
	q-value.
	\newblock \emph{The Annals of Statistics\/}~\textbf{31},~6, 2013--2035.
	
	\bibitem{sun2015etal}
	\textsc{Sun, W.}, \textsc{Reich, B.~J.}, \textsc{Tony~Cai, T.},
	\textsc{Guindani, M.}, \textsc{and} \textsc{Schwartzman, A.} (2015).
	\newblock False discovery control in large-scale spatial multiple testing.
	\newblock \emph{Journal of the Royal Statistical Society: Series B (Statistical
		Methodology)\/}~\textbf{77},~1, 59--83.
	
	\bibitem{winkler2014}
	\textsc{Winkler, A.~M.}, \textsc{Ridgway, G.~R.}, \textsc{Webster, M.~A.},
	\textsc{Smith, S.~M.}, \textsc{and} \textsc{Nichols, T.~E.} (2014).
	\newblock Permutation inference for the general linear model.
	\newblock \emph{Neuroimage\/}~\emph{92}, 381--397.
	
\end{thebibliography}

\end{document}